\documentclass[twoside, 11pt, preprint]{article}

\usepackage{jmlr2e}
\usepackage{amsmath, amsfonts}
\usepackage[ruled, vlined, linesnumbered]{algorithm2e}
\usepackage{eucal}
\usepackage{bm}
\usepackage{colortbl}
\usepackage{booktabs}
\usepackage{float}
\usepackage{siunitx}
\usepackage{multirow}
\usepackage{tikz}
\usepackage{pgfplots}
\usepackage{setspace}
\usepackage{caption}
\usepackage{subcaption}
\usepackage{hyperref}

\usepackage{lastpage}
\jmlrheading{1}{2024}{1-\pageref{LastPage}}{00/24; Revised 00/24}{00/24}{00-0000}{Bryan Andrews and Erich Kummerfeld}

\ShortHeadings{Better Simulations}{Andrews and Kummerfeld}
\firstpageno{1}

\usetikzlibrary{arrows.meta, positioning}
\pgfplotsset{compat=1.8}
\usepgfplotslibrary{groupplots}
\usepgfplotslibrary{fillbetween}

\definecolor{color0}{HTML}{636EFA}
\definecolor{color1}{HTML}{EF553B}
\definecolor{color2}{HTML}{00CC96}
\definecolor{color3}{HTML}{AB63FA}
\definecolor{color4}{HTML}{FFA15A}
\definecolor{color5}{HTML}{19D3F3}
\definecolor{color6}{HTML}{FF6692}
\definecolor{color7}{HTML}{B6E880}
\definecolor{color8}{HTML}{FF97FF}
\definecolor{color9}{HTML}{FECB52}

\let \mc   = \mathcal

\let \mt   = \mathtt
\let \mbb  = \mathbb
\let \tu   = \textup

\let \eq   = \equiv

\let \sube = \subseteq
\let \sm   = \setminus
\let \es   = \varnothing
\let \ta   = \rightarrow
\let \at   = \leftarrow

\newcommand{\pa}[2]{\ensuremath{\textrm{\tu{pa}}_{#1}(#2)}}
\newcommand{\ch}[2]{\ensuremath{\textrm{\tu{ch}}_{#1}(#2)}}
\newcommand{\s}[1][0.5]{\ensuremath{\mkern #1 mu}}

\newcommand{\ind}{{\;\perp\!\!\!\perp\;}}

\newtheorem{conv}{Convention}

\begin{document}

\title{Better Simulations for Validating Causal Discovery with the DAG-Adaptation of the Onion Method}

\author{\name Bryan Andrews \email andr1017@umn.edu \\
        \addr Department of Psychiatry \& Behavioral Sciences, University of Minnesota
        \AND
        \name Erich Kummerfeld \email erichk@umn.edu \\
        \addr Institute for Health Informatics, University of Minnesota}

\editor{}

\maketitle

\begin{abstract} 
    The number of artificial intelligence algorithms for learning causal models from data is growing rapidly. Most ``causal discovery'' or ``causal structure learning'' algorithms are primarily validated through simulation studies. However, no widely accepted simulation standards exist and publications often report conflicting performance statistics---even when only considering publications that simulate data from linear models. In response, several manuscripts have criticized a popular simulation design for validating algorithms in the linear case.

    We propose a new simulation design for generating linear models for directed acyclic graphs (DAGs): the DAG-adaptation of the Onion (DaO) method. DaO simulations are fundamentally different from existing simulations because they prioritize the distribution of correlation matrices rather than the distribution of linear effects. Specifically, the DaO method uniformly samples the space of all correlation matrices consistent with (i.e. Markov to) a DAG. We also discuss how to sample DAGs and present methods for generating DAGs with scale-free in-degree or out-degree. We compare the DaO method against two alternative simulation designs and provide implementations of the DaO method in Python and R: \url{https://github.com/bja43/DaO_simulation}. We advocate for others to adopt DaO simulations as a fair universal benchmark.
\end{abstract}


\section{Introduction}
\label{sec:introduction}

    Learning causal relationships is central to many areas of scientific research, including medicine, climate science, economics, and psychology. Due to the difficulties of setting up proper randomized experiments and the growth in available non-experimental datasets, causal relationships are increasingly being learned from observational data \citep{Pearl2020-xc}. Another factor contributing to this trend is the increasing volume of new causal discovery algorithms (CDAs) being published every year \citep{Spirtes1993-od, Spirtes2010-cg, Eberhardt2017-wi, Malinsky2018-mi, Glymour2019-ce}. CDAs output graphs intended to accurately represent causal relationships between the variables in an input dataset, but lack finite-sample guarantees without unrealistic assumptions \citep{kalisch2007estimating, Spirtes2014-au, uhler2013geometry, pmlr-v177-wang22a}. Moreover, nearly all observational datasets lack a comprehensive causal ground truth which makes simulation studies the primary means of CDA validation \citep{Ramsey2017-rt,Ramsey2020-wn,rios2021benchpress,kummerfeld2023power}.

    Unfortunately, CDA performance can heavily depend on simulation design and the field currently has no simulation standards. Without field standards: (i) algorithm developers can cherry-pick simulations that emphasize their method's strengths or their competitor's weaknesses and (ii) arbitrarily designed simulations may randomly favor certain CDAs over others for unknown reasons. As a consequence, several manuscripts warn that careless simulation practices have biased many well-cited studies \citep{reisach2021beware, kaiser2022unsuitability, reisach2023simple, ng2024structure}.
    
    In this paper, we present a CDA simulation design as a potential field standard for the linear case: the DAG-adaptation of the Onion (DaO) method. The DaO method takes as input a directed acyclic graph (DAG) $\mc G$, and outputs a correlation matrix $R$. The matrix $R$ is sampled uniformly at random over the space of all correlation matrices satisfying the Markov property for $\mc G$. The specific approach we use is adapted from the Onion method of \cite{ghosh2003behavior, ghosh2009corrigendum} by modifying it for DAG models. The matrix $R$ and $\mc G$ define a parametric model that we can use to efficiently generate data. The result of running CDAs on the simulated data can be compared to $\mc G$ to determine their performance.
    
    The DaO method: (i) avoids prioritizing any correlation matrix over another, (ii) ensures that all possible correlation matrices are represented, (iii) does not have tuning parameters that could be used to selectively benefit some CDAs over others, and (iv) offers a fair method for evaluating CDAs. These benefits necessarily come at a cost: the DaO method produces a distribution of models that may not reflect real-world causal systems. For example, data generated from DaO simulations do not resemble brain imaging data, climate science data, omics data, or electronic health record data. The DaO method is intended for generic domain-free simulations. This method could plausibly be extended to specific data domains in the future, but such an extension is outside the scope of the current paper.

    To provide a complete simulation procedure, we also present a method for generating DAGs with scale-free in-degree or out-degree. The generated distribution of DAGs aligns with the classic scale-free graph distribution of Price's model. The implementation of these methods starts with Erd\H{o}s R{\'e}nyi graph sampling and adds no additional tuning parameters. Combining this with the DaO method creates a simulation study design with no tuning parameters beyond those needed for Erd\H{o}s R{\'e}nyi graph sampling. We advise that future CDA simulation studies use the design we present here as a uniform and fair benchmark.

    \textbf{Paper Organization.} We first summarize novel contributions of this paper, briefly consider the motivation for uniformly sampling the space of correlation matrices, and briefly review previous work on problematic simulation study designs for validating CDAs. Section \ref{sec:background} provides some necessary notation, definitions, and other background information. Section \ref{sec:methods} presents our methods for sampling DAGs, the DaO method, and relevant correctness proofs. Section \ref{sec:eval_of_params} presents an empirical evaluation that (i) compares models produced by the DaO method against models produced by two other simulation designs, (ii) demonstrates pitfalls encountered by (at least one of) the other simulation designs, and (iii) demonstrates that the DaO method avoids these pitfalls. Section \ref{sec:sim_study} demonstrates how a variety of CDAs perform on DaO simulations compared to the other two simulation designs. The paper concludes with a brief discussion in Section \ref{sec:discussion}.

\subsection{Contributions}
\label{sec:contributions}

    Our novel contributions include the following:
    \begin{enumerate}
        \item We present a new method, the DAG-adaptation of the Onion (DaO) method, for simulating data from a given DAG. The DaO method avoids several known problems with standard DAG model simulation designs.
        \item We prove that the DaO method randomly samples correlation matrices uniformly over the space of all correlation matrices that satisfy the Markov property for the provided DAG.
        \item We present a new method for rewiring a DAG to make the in-degree or out-degree (or both) scale-free according to Price's model.
        \item We demonstrate advantages of the proposed study design over alternative designs.
        \item We provide a novel explanation for why the simulation design used by many continuous optimization CDA publications report performance statistics in conflict with those from other simulation designs.
    \end{enumerate}

\subsection{Previous Work}
\label{sec:previous_work}

    \cite{reisach2021beware} introduced the concept of \emph{varsortability}. They observed that in the Zheng, Aragam, Ravikuma, and Xing (ZARX) simulation \citep{zheng2018dags} (a common method for simulating data from a DAG which is described in more detail in Section \ref{sec:eval_of_params}) the causal order of the variables correlates with the marginal variance of the variables. This is problematic because the marginal variance of a variable is scale dependent, making it an arbitrary quantity that should contain no information about the DAG. However, since the ZARX simulation design produces data with varsortability, marginal variance could in principle be used to gain information about the topological order of the DAG. Many CDAs are unaffected by marginal variance and thus would not make use of this information. Some CDAs, however, appear to be responsive to marginal variance. Such algorithms would have unrealistically strong performance in simulation studies with varsortability.

    While varsortability can be removed by standardizing the data, \cite{reisach2023simple} show in a followup paper that the ZARX simulation design has another sortability property that is not affected by rescaling the data: \emph{$R^2$-sortability}. For $R^2$-sortability, the statistic of interest is, ``the fraction of a variable's variance explained by all others, as captured by the coefficient of determination $R^2$.'' Many simulation designs produce data such that a variable's ranking based on $R^2$ values is associated with its ranking in the DAG's topological order. Whether real world causal systems also exhibit $R^2$-sortability is an empirical question, however it is certainly plausible that common simulation designs have significantly more $R^2$-sortability than is realistic or reasonable.

    These sortability concepts have also been discussed in relation to a recently developed class of DAG-learning algorithms that use continuous optimization to solve a differentiable loss function. In the first paper to actively criticize continuous optimization for learning DAGs, \cite{kaiser2022unsuitability} focus on the Non-combinatorial Optimization via Trace Exponential and Augmented lagRangian for Structure learning (NOTEARS) method \citep{zheng2018dags}. They demonstrate through examples and theory that NOTEARS is not scale-invariant. In other words, the output of NOTEARS depends on the units of measurement for the variables in the dataset. This suggests that NOTEARS makes use of varsortability to achieve good performance. They also show that when the simulation design removes or reverses the varsortability, the performance of those methods can drop dramatically.

    \cite{ng2024structure} also investigated how the apparent performance of continuous optimization DAG-learning algorithms depends on simulation design. While they are critical of methods such as NOTEARS, they point to different issues than those raised by Kaiser and Reisach. In terms of simulation design, they indicate the use of equal versus non-equal variances in the simulated independent error terms as being responsible for the unusual simulation results \citep{peters2014identifiability}.

    These previous works highlight a general problem with prior simulation designs for evaluating DAG-learning algorithms. By producing data distributions that associate irrelevant statistics with the DAG's structure, these designs can artificially skew results. In this paper, we present the DAG-adaptation of the Onion (DaO) method, which avoids these issues without introducing tuning parameters that might induce other problems. The DaO method accomplishes this by uniformly sampling the space of correlation matrices consistent with (i.e. that satisfy the Markov property with respect to) a DAG.

\subsection{Why Simulate from DAGs by Sampling Correlation Matrices Uniformly?}
\label{sec:why_uniform}

    The benefits of the DaO method fall into two categories: (1) advantages of directly sampling from the space of correlation matrices; (2) advantages of sampling uniformly at random. Throughout this section, it is assumed that we wish to simulate data from some known DAG.

\subsubsection{Advantages of Directly Sampling the Correlation Matrix}

    
    Existing linear simulation designs almost exclusively sample from a known distribution of edge weights and independent noise terms, which in turn imply a (typically unknown or uncharacterized) distribution of correlation matrices. The DaO method takes the opposite approach: we prioritize sampling from a known distribution over the space of correlation matrices, and the sampled matrix implies a standardized set of edge weights and independent noise terms. This has some advantages, including the following.
    
    First, a correlation matrix and sample size are sufficient input for many CDAs. As such, for a given DAG, it makes sense to consider how CDAs perform across the space of possible (i.e. Markov) correlation matrices. By prioritizing our control over the distribution of correlation matrices, we can directly interrogate performance across the space of possible CDA inputs.
    
    Second, directly sampling correlation matrices will immediately and trivially prevent simulation artifacts like varsortability. A correlation matrix characterizes a standardized distribution where all variables have equal variance, making varsortability impossible.
    
    %
    
    Third, the parametric models produced by directly sampling a correlation matrix are inherently standardized. Almost all other simulations sample edge weights and independent noise terms in a way that does not produce a standardized model. This can lead to situations where effect sizes (in terms of signal versus noise) are much stronger or weaker than one might expect, potentially even creating near-deterministic relationships; see Figures \ref{fig:er_sims}, \ref{fig:sfi_sims}, and \ref{fig:sfo_sims}. It is possible to directly sample standardized edge weights and independent noise terms, however the only existing method we are aware of that meets this criterion has several restrictions which the DaO method does not have \citep{kummerfeld2023power}. For example, that method requires all edges in a model to have the same strength---DaO simulations have no such restrictions.

\subsubsection{Advantages of Sampling Uniformly}

    
    After establishing a paradigm of directly sampling correlation matrices, it is intuitive to sample uniformly from this space. Uniform sampling across correlation matrices also has several benefits, including the following.
    
    First, sampling from the uniform distribution of correlation matrices means there are no simulation parameters or other design choices. This is convenient, prevents users from cherry-picking simulation parameters, and ensures that simulations from different studies are consistent.
    
    Second, uniform sampling across all correlation matrices means that no correlation matrix is missed. All previous simulation designs that we are aware of omit substantial portions of the possible correlation matrices, and as such provide no findings about how CDAs perform on parts of the problem space. This could also be accomplished by another distribution with total support.
    
    Third, for assessment of learning tasks it is common to assume a generic condition where we lack knowledge about what to expect. When evaluating CDAs in a generic context (i.e. where we do not have a precise topic that the CDA is intended for), it is hard to say that having good performance on any particular correlation matrix is more important than another. This lack of knowledge in a generic context is accurately captured by a uniform distribution.
    
    Fourth, sampling uniformly from the space of all correlation matrices is fair and embodies a valuable benchmark. The CDAs that perform best in these simulations will be those that have the best average performance across all inputs. This is a desirable property for any general purpose CDA to have, and is a fair test that puts all CDAs on an equal playing field.

\section{Background}
\label{sec:background}

    This section provides background on the relevant notation, models, and theory. The topics covered include directed acyclic graphs (DAGs), structural equation models (SEMs), the original Onion method, the multivariate Pearson type II distribution, and the Erd\H{o}s R{\'e}nyi (ER) method for generating DAGs.

\subsection{Directed Acyclic Graphs}
\label{sec:dags}

    A \textit{directed acyclic graph} (DAG) is an ordered pair $\mc G = (V, E)$ consisting of a finite vertex set $V$ and an edge set of ordered pairs $E \sube V \times V$ with no \textit{directed cycles}.\footnote{A \textit{directed cycle} is a vertex sequence $\langle v_1, \dots, v_k \rangle$ $(k > 1)$ where $v_1 = v_k$ and $(v_i, v_{i+1}) \in E$ for all $i < k$.} The vertices connected to vertex $i \in V$ by an edge are described as \textit{parents} and \textit{children} respectively:
    \[
        \pa{\mc G}{i} \eq \{ j \in V \; : \; (i, j) \in E \}
        \s[72]
        \ch{\mc G}{i} \eq \{ j \in V \; : \; (j, i) \in E \}.
    \]
    We relate $\mc G$ to a \textit{total order} $\prec$ on $V$ using two conditions.\footnote{A \textit{total order} is a binary relation on $V$ that is reflexive, transitive, antisymmetric, and strongly connected.}
    \begin{itemize}
        \item The total order $\prec$ is \textit{consistent} if for all $i, j \in V$: 
        \[
            j \in \pa{\mc G}{i}  \s[18] \Rightarrow \s[18]  j \prec i.
        \]
        \item The total order $\prec$ \textit{source-first} if for all $i, j \in V$:
        \[
            |\pa{\mc G}{i}| \neq 0 \s[12] \tu{and} \s[12] |\pa{\mc G}{j}| = 0 \s[18] \Rightarrow \s[18] j \prec i.
        \]
    \end{itemize}
    
    
    
    
    
    



    \noindent Let $[ \, p \, ] = \{ i \in \mbb N \; : \; 1 \leq i \leq p \}$ and note that $\prec$ induces a bijection $\pi: [ \, p \, ] \rightarrow V$ between the first $p$ natural numbers and the elements of $V$. The following convention will be used throughout the paper.
    \begin{conv}
        \label{conv:1}
        Abusing notation, we use $i \in [ \, p \, ]$ to denote $\pi(i)$ and $i \in V$ to denote $\pi^{-1}(i)$. For any graph with vertex set $V$ let $\prec$ denote a source-first consistent order. This can be done on a graph by graph basis and therefore does not put any restriction on the graphs we consider. Using the established abuse of notation, this implies $1 \prec \dots \prec p$ is a source-first consistent order and $i \prec j$ if and only if $i < j$ for all $i, j \in V$. 
    \end{conv}

\subsection{DAG Models}
\label{sec:dag_models}

    DAG models are probabilistic models whose conditional independence relationships are explicitly represented by a DAG. In particular, DAGs graphically encode conditional independence relationships between their vertices which can be read off using $d$-separation \citep{pearl1988probabilistic, verma1990causal}. Markov properties define subsets of these relationships---usually by a simple graphical criterion---that are sufficient to imply the full set of conditional independence relationships encoded by the DAG.
    
    The \textit{ordered Markov property} does so using a consistent total order \citep{lauritzen1990independence}.\footnote{\cite{lauritzen1990independence} originally called this Markov property the well-numbering Markov property.} Accepting Convention \ref{conv:1}, a probability distribution over $V$ satisfies the ordered Markov property with respect to a DAG $\mc G = (V, E)$ if and only if:
    \[
        i \ind [ \, i - 1 \, ] \sm \pa{\mc G}{i} \mid \pa{\mc G}{i} \s[12] \tu{for all} \s[12] i \in V
    \]
    where the ternary relation $A \ind B \mid C$ for disjoint sets $A, B, C \sube V$ denotes that $A$ and $B$ are independent conditioned on $C$ \citep{dawid1979conditional}. 

    When a distribution satisfies the ordered Markov property with to respect $\mc G$, we say the distribution and its corresponding parametric model are \textit{Markov to} $\mc G$. The Faithfulness condition compliments the ordered Markov property---it requires that the full set of conditional dependence (rather than independence) relationships encoded by the DAG are contained in the model. When a distribution contains the conditional dependence relationships encoded by $\mc G$, we say the distribution and its corresponding parametric model are \textit{faithful to} $\mc G$.
    
    Recursive structural equation models (SEMs) characterize a family of DAG models parameterized (primarily) by covariance/correlation.\footnote{General SEMs can model latent variables and feedback. These models can have more complex independence structures not compatible with DAGs \citep{bollen1989structural, drton2018algebraic}.} In these models, conditional independence is equivalent to partial correlation.






\subsection{Recursive Structural Equation Models}
\label{sec:rsems}

    A \textit{recursive SEM} is a system of equations with no feedback loops or correlation error terms \citep{bollen1989structural}. Let $X = \{ X_i \; : \; i \in V \}$ and $Z = \{ Z_i \; : \; i \in V \}$ be collections of random variables indexed by $V$ where $X$ denotes model variables and $Z$ denotes error terms. $Z$ is distributed according to $\mc E(\Omega)$ where $\mc E$ is a distribution with finite covariance represent by matrix $\Omega$.
    
    Let $\mbb R^{V \times V}$ be the space of $|V| \times |V|$ dimensional real-valued matrices, and let $\mbb S^{V}_{++}$ be the space of $|V| \times |V|$ dimensional positive definite matrices. Accepting Convention \ref{conv:1}, a recursive SEM for $\mc G$ is a distribution characterized by the system of equations:

    \begin{equation}
        \label{eq:sem}
        X = X B^\top + Z \s[72] Z \sim \mc E(\Omega)
    \end{equation}
    where
    \begin{align}
        \label{eq:bspace} B &\in \{ (\beta_{i,j}) \in \mbb R^{V \times V} \; : \; \beta_{i,j} = 0 \s[8] \tu{if} \s[8] j \not \in \pa{\mc G}{i} \}; \\[2mm]
        \label{eq:ospace} \Omega &\in \{ (\omega_{i,j}) \in \mbb S^{V}_{++} \; : \; \omega_{i,j} = 0 \s[8] \tu{if} \s[8] i \neq j \}.
    \end{align}
    
    Let $I$ denote the $p \times p$ identity matrix. $I - B$ is lower triangular and invertible---the linear system in Equation \ref{eq:sem} is solved uniquely by $X = Z \, (I - B)^{-\top}$ and has covariance matrix:
    \begin{equation}
        \label{eq:sigma}
        \Sigma = (I - B)^{-\top} \, \Omega \, (I - B)^{-1}.
    \end{equation}

    \noindent In order to establish the connection between recursive SEMs and DAGs, let $J \sube [ \, i - 1 \, ]$ such that $\pa{\mc G}{i} \sube J$:
    \begin{equation}
        \label{eq:beta}
        B_{i,J} = \Sigma_{i,J} \, (\Sigma_{J,J})^{-1}
    \end{equation}
    \begin{equation}
        \label{eq:omega}
        \Omega_{i,i} = \Sigma_{i,i} - \Sigma_{i,J} \, (\Sigma_{J,J})^{-1} \, \Sigma_{J,i}
    \end{equation}
    $J$ is usually defined as $\pa{\mc G}{i}$---the extension we give here is implied by the ordered Markov property.

\subsection{The Onion Method}
\label{sec:onion_method}


    \cite{ghosh2003behavior, ghosh2009corrigendum} proposed the Onion method for uniformly sampling correlation matrices. In other words, the Onion method characterizes the distribution $R \sim f : \{ (\sigma_{i,j}) \in \mbb S^{V}_{++} \; : \; \sigma_{i,j} = 1 \s[8] \tu{if} \s[8] i = j \} \ta \mbb R$ such that:
    \[
        f(R) \propto 1.
    \]
    Their approach sets up a recurrence relation where a correlation matrix $R_{i+1}$ is sampled by appending a vector $r_{i+1}$ to the rows and columns of an existing correlation matrix $R_i$:
    \[
    R_{i+1} =
    \begin{bmatrix}
        R_{i} & r_{i+1} \\
        r_{i+1}^\top & 1
    \end{bmatrix}
    \]


    Ghosh and Henderson call $r_{i+1}$ the completion of $R_{i}$ in $R_{i+1}$ and provide a necessary and sufficient condition on $r_{i+1}$ for $R_{i+1}$ to be positive definite---we use this condition as well. Lemma 4.3 from \cite{ghosh2003behavior}:
    \begin{lemma}
        \label{lem:r_pd}
        If $R_i$ is positive definite, then $R_{i+1}$ is positive definite if and only if:
        \[
            1 - r_{i+1}^\top R_i^{-1} r_{i+1} > 0.
        \]
    \end{lemma}
    
    Ghosh and Henderson named their method the Onion method because it builds up a correlation matrix one layer at a time. \cite{lewandowski2009generating} reformulated the Onion method in terms of elliptical distributions. Their reformulation relies on the multivariate Pearson type II distribution---we rely on this distribution as well.

\subsection{Multivariate Pearson Type II}
\label{sec:mpii}

    $W \sim \textit{mPII}_d(\gamma)$ follows a multivariate Pearson type II distribution if it has density:\footnote{The multivariate Pearson type II distribution is usually defined: $f(w \mid a) \propto (1 - w^\top w)^{a - 1}$. We use $\gamma = a - \frac{1}{2}$ to simplify calculations in Section \ref{sec:go2dag}.}
    \begin{equation}
        \label{eq:mpii}
        f(w \mid \gamma) \propto (1 - w^\top w)^{\gamma - \frac{1}{2}}
    \end{equation}
    where $\gamma > -\frac{1}{2}$ and $w \in \mbb R^d$ such that $w^\top w < 1$. Alternatively, $W = Q^\frac{1}{2} U$ where:
    \begin{align*}
        Q &\sim \textit{beta}(\tfrac{d}{2}, \gamma + \tfrac{1}{2}); \\
        U &\sim \textit{uniform}(\{ u \in \mbb R^d \; : \; \Vert u \Vert_2 = 1 \}).
    \end{align*}
    In the alternative formulation, $Q$ is distributed uniformly on the surface of the unit $d$-sphere. For more details see \cite{fang1990symmetric, khalafi2014multivariate}.
    

\subsection{Sampling Random DAGs}
\label{sec:sampling_dags}

    We consider two approaches for generating graphs: Erd\H{o}s R{\'e}nyi (ER) where edges occur with equal probability \citep{erdos59random} and Scale-Free (SF) where edges occur such that the distribution of vertex degrees follows a power-law \citep{barabasi1999emergence}. These methods produce undirected graphs rather than DAGs, but are frequently adapted to generate DAGs by imposing a topological order on the vertices---we also take this approach. Algorithm \ref{alg:er_dag} contains pseudo-code for sampling ER-DAGs using the total order established by Convention \ref{conv:1}.

    \begin{algorithm}[H]
    \DontPrintSemicolon
    \setstretch{1.1}
    \caption{$\mt{ER\tu{-}DAG}(V, \alpha)$}
    \label{alg:er_dag}
    \KwIn{$\mt{vertices}: V \s[24] \mt{avg\tu{-}deg}: \alpha$}
    \KwOut{$\mt{DAG}: \mc G = (V, E)$}
    $E \at \es$ \;
    \Repeat{$|E| = \frac{\alpha}{2} \s[3] |V|$}{
        $(i, j) \sim f : V \s[-3] \times \s[-1] V \ta \mbb R \s[12] \tu{s.t.} \s[12] f \propto 1_{j < i}$ \;
        $E \at E \cup \{ (i, j) \}$ \;
    }
\end{algorithm}

    In Algorithm \ref{alg:er_dag}, $1_{j<i}$ denotes the indicator function for $j < i$. Algorithm \ref{alg:er_dag} ($\mt{ER\tu{-}DAG}$) takes vertex set $V$ and average degree $\alpha$ as input to construct a DAG. The total order over $V$ established by Convention \ref{conv:1} is imposed on the vertices to direct the edges. Edges consistent with the imposed order are sampled uniformly at random until an average degree of $\alpha$ is achieved. Since the operation of some causal discovery algorithms (CDAs) can depend on the order of the variables, the vertex order of a DAG can (and perhaps should) be randomized after sampling.
    
    Notably, $\mt{ER\tu{-}DAG}$ is biased and does not produce DAGs uniformly at random. Table \ref{tab:er_dags} gives a simple counterexample. Each cell in Table \ref{tab:er_dags} shows a DAG with three vertices and two edges. The rows and columns enumerate the possible undirected graphs and topological orders, respectively. Assuming the topological order is chosen arbitrarily, sampling a 3-vertex 2-edge DAG according to $\mt{ER\tu{-}DAG}$ is equivalent to selecting a row and column uniformly at random. Some DAGs occur only once in this table, however, others occur in two cells. As such, some DAGs (the darkly-shaded cells) will be sampled twice as often as others (the lightly-shaded cells).


    \begin{table}[ht!]
    \centering
    \begin{tabular}{ccccccc}
        \toprule
        & $1 < 2 < 3$ & $1 < 3 < 2$ & $2 < 1 < 3$ & $2 < 3 < 1$ & $3 < 1 < 2$ & $3 < 2 < 1$ \\
        \cmidrule(lr){2-2} \cmidrule(lr){3-3} \cmidrule(lr){4-4} \cmidrule(lr){5-5} \cmidrule(lr){6-6} \cmidrule(lr){7-7}
        \includegraphics[page=1, scale=0.8]{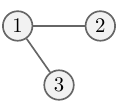} & \cellcolor{gray!20} \includegraphics[page=2, scale=0.8]{images/figs.pdf} & \cellcolor{gray!20} \includegraphics[page=2, scale=0.8]{images/figs.pdf} & \cellcolor{gray!10} \includegraphics[page=3, scale=0.8]{images/figs.pdf} & \cellcolor{gray!20} \includegraphics[page=4, scale=0.8]{images/figs.pdf} & \cellcolor{gray!10} \includegraphics[page=5, scale=0.8]{images/figs.pdf} & \cellcolor{gray!20} \includegraphics[page=4, scale=0.8]{images/figs.pdf} \\
        \includegraphics[page=6, scale=0.8]{images/figs.pdf} & \cellcolor{gray!10} \includegraphics[page=7, scale=0.8]{images/figs.pdf} & \cellcolor{gray!20} \includegraphics[page=10, scale=0.8]{images/figs.pdf} & \cellcolor{gray!20} \includegraphics[page=8, scale=0.8]{images/figs.pdf} & \cellcolor{gray!20} \includegraphics[page=8, scale=0.8]{images/figs.pdf} & \cellcolor{gray!20} \includegraphics[page=10, scale=0.8]{images/figs.pdf} & \cellcolor{gray!10} \includegraphics[page=9, scale=0.8]{images/figs.pdf} \\
        \includegraphics[page=11, scale=0.8]{images/figs.pdf} & \cellcolor{gray!20} \includegraphics[page=12, scale=0.8]{images/figs.pdf} & \cellcolor{gray!10} \includegraphics[page=15, scale=0.8]{images/figs.pdf} & \cellcolor{gray!20} \includegraphics[page=12, scale=0.8]{images/figs.pdf} & \cellcolor{gray!10} \includegraphics[page=13, scale=0.8]{images/figs.pdf} & \cellcolor{gray!20} \includegraphics[page=14, scale=0.8]{images/figs.pdf} & \cellcolor{gray!20} \includegraphics[page=14, scale=0.8]{images/figs.pdf} \\
        \bottomrule
    \end{tabular}
    \caption{The typical method for sampling ER-DAGs is not uniform.}
    \label{tab:er_dags}
\end{table}

    Alternatively, \cite{melanccon2001random} proposed a method to sample DAGs uniformly at random. However, their procedure relies on MCMC sampling and is not scalable to hundreds of variables. Moreover, their approach can only constrain the sample space by specifying an upper bound on the number of edges, rather than something more precise such as an exact average degree. 
    
    SF-DAGs are commonly generated following \cite{price1976general}. Many studies use the $\mt{igraph}$ C library \citep{igraph} which by default generates SF-DAGs with scale-free in-degree and constant vertex out-degree (equal to $\frac{\alpha}{2}$). This approach can generate SF-DAGs with scale-free out-degree and constant vertex in-degree by reversing the direction of the edges in the graph. However, it is not straightforward to add variation to the non-scale-free vertex degree or to generate DAGs that have scale-free in-degree and out-degree at the same time.

\section{Methods}
\label{sec:methods}

    This section presents the primary content of this paper. First, we present methods for randomly rewiring the edges of DAGs to have scale-free in-degree or out-degree (or both). We then present the DaO method for sampling correlation matrices that are Markov to a given DAG. This includes theoretical results that the distribution of matrices produced by the DaO method uniformly samples the space of all correlation matrices Markov to the input DAG.

\subsection{Scale-free Randomly Rewiring DAGs}
\label{sec:rewiring_dags}

    Unlike the DAG sampling method described in Section \ref{sec:sampling_dags}, the methods we present here are for modifying an existing DAG in a way that retains some properties, such as average degree and number of vertices, while changing other properties, such as having a power-law distribution for the vertex in-degree or out-degree.

    



    \noindent
    \begin{minipage}{0.48\textwidth} \begin{algorithm}[H]
    \DontPrintSemicolon
    \setstretch{1.1}
    \caption{$\mt{SFi\tu{-}DAG}(\mc H)$}
    \label{alg:sfi_dag}
    \KwIn{$\mt{DAG}: \mc H = (V, \, \cdot \,) \; \tu{s.t.} \; |V| = p$}
    \KwOut{$\mt{DAG}: \mc G = (V, E)$}
    $E \at \es$ \;
    \For{$i \at p$ \KwTo $1$}{
        \Repeat{$|\ch{\mc G}{i}| = |\ch{\mc H}{i}|$}{ 
            $j \sim f : V \ta \mbb R \s[12] \tu{s.t.}$ $f \propto 1_{i \, < \, j} + |\pa{\mc G}{j}|$ \;
            $E \at E \cup \{ (j, i) \}$ \;
        }
    }
\end{algorithm} \end{minipage}
    \begin{minipage}{0.01\textwidth} \hfill \end{minipage}
    \begin{minipage}{0.48\textwidth} \begin{algorithm}[H]
    \DontPrintSemicolon
    \setstretch{1.1}
    \caption{$\mt{SFo\tu{-}DAG}(\mc H)$}
    \label{alg:sfo_dag}
    \KwIn{$\mt{DAG}: \mc H = (V, \, \cdot \,) \; \tu{s.t.} \; |V| = p$}
    \KwOut{$\mt{DAG}: \mc G = (V, E)$}
    $E \at \es$ \;
    \For{$i \at 1$ \KwTo $p$}{
        \Repeat{$|\pa{\mc G}{i}| = |\pa{\mc H}{i}|$}{ 
            $j \sim f : V \ta \mbb R \s[12] \tu{s.t.}$ $f \propto 1_{j \, < \, i} + |\ch{\mc G}{j}|$ \;
            $E \at E \cup \{ (i, j) \}$ \;
        }
    }
\end{algorithm} \end{minipage}

    In Algorithms \ref{alg:sfi_dag}, and \ref{alg:sfo_dag}: $1_{i<j}$ and $1_{j<i}$ denote the indicator function for $i < j$ and $j < i$ respectively. Algorithm \ref{alg:sfi_dag} ($\mt{SFi\tu{-}DAG}$) resamples the parents of each vertex via preferential attachment in reverse topological order resulting in a scale-free in-degree distribution while maintaining the prior out-degree distribution. Conversely, Algorithm \ref{alg:sfo_dag} ($\mt{SFo\tu{-}DAG}$) resamples the children of each vertex via preferential attachment in topological order resulting in a scale-free out-degree distribution while maintaining the prior in-degree distribution. 
    
    These algorithms allow one to provide a DAG that has some desirable properties, such as number of edges or in/out-degree distribution, and modify it so that either the in-degree or the out-degree is scale-free. Additionally, both procedures may be applied to produce a graph with both scale-free in-degree and out-degree.\footnote{Applying $\mt{SFi\tu{-}DAG}$ before $\mt{SFo\tu{-}DAG}$ versus $\mt{SFo\tu{-}DAG}$ before $\mt{SFi\tu{-}DAG}$ results in different in/out-degree distributions. Their application may be iterated in an alternating manner if identical in/out-degree distributions are desired.} Figures \ref{fig:degs} and \ref{fig:adj_mats} illustrate in/out-degree distributions for DAGs generated by $\mt{SFi\tu{-}DAG}$ and $\mt{SFo\tu{-}DAG}$ compared $\mt{ER\tu{-}DAG}$.
    
    \begin{figure}[ht]
        \centering
        \includegraphics[scale=0.6]{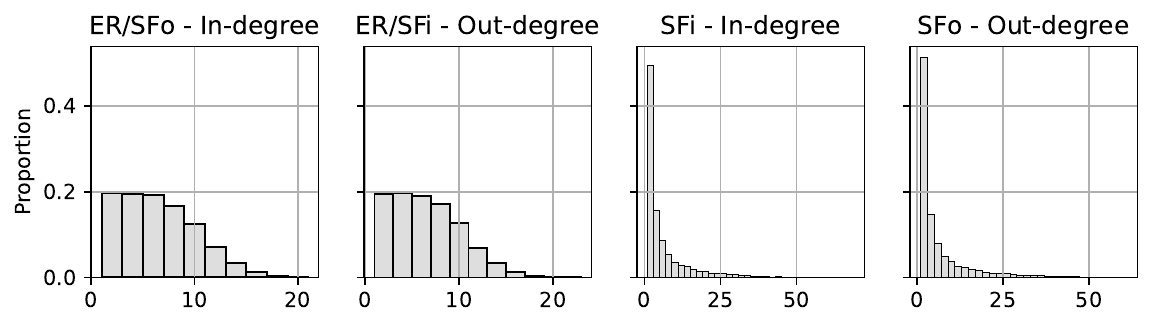}
        \caption{Vertex in/out-degree distributions for 100 DAGs with $|V| = 100$ and $\alpha = 10$.}
        \label{fig:degs}
    \end{figure}

    \begin{figure}[ht]
        \centering
        \begin{subfigure}[b]{0.3\textwidth}
            \raggedleft
            \includegraphics[width=0.5\textwidth]{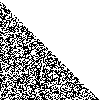}
        \end{subfigure}
        \hfill
        \begin{subfigure}[b]{0.3\textwidth}
            \centering
            \includegraphics[width=0.5\textwidth]{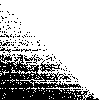}
        \end{subfigure}
        \hfill
        \begin{subfigure}[b]{0.3\textwidth}
            \raggedright
            \includegraphics[width=0.5\textwidth]{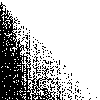}
        \end{subfigure}
        \caption{Edge matrices for ER/SFi/SFo-DAGs with $|V| = 100$ and $\alpha = \frac{99}{2}$ (density $\frac{1}{2}$).}
        \label{fig:adj_mats}
    \end{figure}

\subsection{Generalizing the Onion Method to DAGs}
\label{sec:go2dag}

    We first present an intuition and the theory for the DAG-adaptiation of the Onion (DaO) method. The main result is proven by Theorem \ref{thm:dao}. Pseudocode for the DaO method is provided in Section \ref{sec:dao}.

    To create the DaO method, we adapt the Onion method to sample uniformly from the space of correlation matrices Markov to a given DAG. In addition to sampling the correlation matrix, we keep track of the corresponding $B$ and $\Omega$ matrices, which we sample in tandem with the correlation matrix. 
    
    Let $\mc G = (V, E)$ be a DAG and accept Convention \ref{conv:1} so that $V = 1, \dots, p$ and $1 < \dots < p$ is a source first consistent order. Starting from the $m$ dimensional identity matrix where $m = | \{ i \in V \; : \; \pa{\mc G}{i} = \es \} |$, the DAG-adaptation of the onion (DaO) builds up a correlation matrix (and corresponding $B$ and $\Omega$) by adding one layer at a time:
    \[
        R_{i+1} =
        \begin{bmatrix}
            R_{i} & r_{i+1} \\
            r_{i+1}^\top & 1
        \end{bmatrix}
        \s[72]
        B_{i+1} =
        \begin{bmatrix}
            B_{i} & 0_i \\
            b_{i+1}^\top & 0
        \end{bmatrix}
        \s[72]
        \Omega_{i+1} =
        \begin{bmatrix}
            \Omega_{i} & 0_i \\
            0_i^\top & \omega_{i+1}
        \end{bmatrix}
    \]
    where $0_i$ denotes a column vector of $i$ zeros. Intuitively, the non-parents of $i+1$ correspond to zero entries in $b^\top_{i+1}$. Following Ghosh and Henderson, we call $b_{i+1}^\top$ the completion of $B_i$ in $B_{i+1}$ and $\omega_{i+1}$ the completion of $\Omega_i$ in $\Omega_{i+1}$. By Equations \ref{eq:beta} and \ref{eq:omega}:
    \begin{align}
        b_{i+1}^\top &= r_{i+1}^\top R_i^{-1} \label{eq:beta_comp} \\
        \omega_{i+1} &= 1 - r_{i+1}^\top R_i^{-1} r_{i+1} \label{eq:omega_comp}
    \end{align}

    \noindent Define $\Phi(\mc G, R_i)$ as the space of (potentially invalid) completions of $R_{i}$ in $R_{i+1}$ where the indices $b^\top_{i+1}$ that do not correspond to parents of $i+1$ are zero. By Equation \ref{eq:beta_comp}:
    \[
        \Phi(\mc G, R_i) \eq \{ r_{i+1} \in \mbb R^i \; : \; (b_{i+1}^\top)_j = (r_{i+1}^\top R_i^{-1})_j = 0 \s[8] \tu{if} \s[8] j \not \in \pa{\mc G}{i+1} \}
    \]
    Define $\Psi(R_i)$ as the space of valid completions of $R_{i}$ in $R_{i+1}$. By Equation \ref{eq:omega_comp} and Lemma \ref{lem:r_pd}:
    \[
        \Psi(R_i) \eq \{ r_{i+1} \in \mbb R^i \; : \; \omega_{i+1} = 1 - r_{i+1}^\top R_i^{-1} r_{i+1} > 0 \}
    \]
    The DaO method iteratively sample $r_{i+1} \in \Phi(\mc G, R_i) \cap \Psi(R_i)$ to construct a correlation matrix $R$ layer by layer, thereby sampling a correlation matrix Markov to the given DAG.

    \vskip 5mm
    
    We temporarily drop subscripts to simplify notation and let $R = R_i$, $r = r_{i+1}$, $b = b_{i+1}$, $\omega = \omega_{i+1}$, and $k = |\pa{\mc G}{i+1}|$. Additionally:
    \begin{itemize}
        \item Let $P$ be a $i \times i$ permutation matrix where $P^\top R P$ permutes the rows and columns of $R$ such that rows and columns corresponding to parents of $i+1$ come first.
        \item Permutation matrices are orthogonal: $P^\top = P^{-1}$ so we can revert a permutation by multiplying by the transpose. The determinant of a permutation matrix has an absolute value of one: $|\det(P)| = 1$.
        \item Let $L L^\top = P^\top R P$ be the Cholesky decomposition so that $R = P L L^\top P^\top$.
        \item $L$ is an $i \times i$ lower triangular matrix with real and positive diagonal elements. $L^{-1}$ exists and is an $i \times i$ lower triangular matrix with real and positive diagonal elements.
        \item Let $w$ be a point in the unit $i$-sphere: $w \in R^i$ such that $w^\top w < 1$ 
    \end{itemize}

    
    \noindent Let $w = \begin{bmatrix} w_1 \\ w_2 \end{bmatrix}$ where $w_1 \sim \textit{mPII}_k$ and $w_2 = 0_{(i-k)}$. Noting that $(1 - w_1^\top w_1) = (1 - w^\top w)$, we can express the density of $w$ in the same form as Equation \ref{eq:mpii}: $f(w \mid \gamma) \propto (1 - w^\top w)^{\gamma - \frac{1}{2}}$. If we apply the change of variables $w = (P L)^{-1} r$, then $r$ has density:
    \begin{equation}
    \begin{split}
        f(r \mid R, \, \gamma) &\propto \left[1 - ((P L)^{-1} r)^\top ((P L)^{-1} r) \right]^{\gamma - \frac{1}{2}} \; \left|\det((P L)^{-1}) \right| \\
        &= \left[1 - (L^{-1} P^\top \, r)^\top L^{-1} P^\top \, r \right]^{\gamma - \frac{1}{2}} \; \left|\det(P L)^{-1} \right| \\
        &= \left[1 - r^\top P L^{-\top} L^{-1} P^\top \, r \right]^{\gamma - \frac{1}{2}} \; \left[\det(P L) \det(L^\top P^\top) \right]^{-\frac{1}{2}} \\
        &= \left[1 - r^\top (P L L^{\top} P^\top)^{-1} r \right]^{\gamma - \frac{1}{2}} \; \det(P L L^\top P^\top)^{-\frac{1}{2}} \\
        &= (1 - r^\top R^{-1} r)^{\gamma - \frac{1}{2}} \; \det(R)^{-\frac{1}{2}} \label{eq:r|R}
    \end{split}
    \end{equation}

    The change of variables: $r = P L w$ and $w = (P L)^{-1} r$ is central to the DaO method. First we prove the linear transformation used in this change of variables is a bijection.  

    \begin{lemma}
        \label{lem:biject}
        If $R = PLL^\top P^\top$ is positive definite, then the linear transformation characterized by $PL$ is a bijection.
    \end{lemma}

    \begin{proof}
        If $R$ is positive definite, then $\det(R) \neq 0$. Furthermore, $\det(R) = \det(PLL^\top P^\top) = \det(PL)^2$ so $\det(PL) \neq 0$. Accordingly, $PL$ is full rank and therefore a bijection.
    \end{proof}

    \noindent The completions defined in Equations \ref{eq:beta_comp} and \ref{eq:omega_comp} can be reformulated in terms of $w$:
    \begin{align}
       \begin{split}
            b^\top &= \s[8] r^\top R^{-1} \s[8] = \s[8] (P L w)^\top (P L L^\top P^\top)^{-1} \\
            &\s[93] = \s[8] w^\top L^\top P^\top P L^{-\top} L^{-1} P^\top \s[8] = \s[8] w^\top L^{-1} P^\top \label{eq:beta_conv}
        \end{split} \\[2mm]
        \begin{split}
            \omega &= \s[8] 1 - r^\top R^{-1} r \s[8] = \s[8] 1 - (P L w)^\top (P L L^\top P^\top)^{-1} (P L w) \\
            &\s[132] = \s[8] 1 - w^\top L^\top P^\top P L^{-\top} L^{-1} P^\top P L w \s[8] = \s[8] 1 - w^\top w \label{eq:omega_conv}
        \end{split}
    \end{align}

    \noindent Let $w = \begin{bmatrix} w_1 \\ w_2 \end{bmatrix}$ where $w_1 \sim \textit{mPII}_k$ and $w_2 = 0_{(i-k)}$. Noting that $(1 - w_1^\top w_1) = (1 - w^\top w)$, we can express the density of $w$ in the same form as Equation \ref{eq:mpii}: $f(w \mid \gamma) \propto (1 - w^\top w)^{\gamma - \frac{1}{2}}$. If we apply the change of variables $w = (P L)^{-1} r$, then $r$ has density:
    \begin{equation}
    \begin{split}
        f(r \mid R, \, \gamma) &\propto \left[1 - ((P L)^{-1} r)^\top ((P L)^{-1} r) \right]^{\gamma - \frac{1}{2}} \; \left|\det((P L)^{-1}) \right| \\
        &= \left[1 - (L^{-1} P^\top \, r)^\top L^{-1} P^\top \, r \right]^{\gamma - \frac{1}{2}} \; \left|\det(P L)^{-1} \right| \\
        &= \left[1 - r^\top P L^{-\top} L^{-1} P^\top \, r \right]^{\gamma - \frac{1}{2}} \; \left[\det(P L) \det(L^\top P^\top) \right]^{-\frac{1}{2}} \\
        &= \left[1 - r^\top (P L L^{\top} P^\top)^{-1} r \right]^{\gamma - \frac{1}{2}} \; \det(P L L^\top P^\top)^{-\frac{1}{2}} \\
        &= (1 - r^\top R^{-1} r)^{\gamma - \frac{1}{2}} \; \det(R)^{-\frac{1}{2}} \label{eq:r|R}
    \end{split}
    \end{equation}

    \vskip 5mm

    \noindent We add subscripts back and redefine $\Phi(\mc G, R_i)$ and $\Psi(R_i)$ using $w$ in Lemmas \ref{lem:dag} and \ref{lem:corr} respectively. Theorem \ref{thm:dao} uses these lemmas to prove that the DaO method samples positive definite correlation matrices that are Markov to $\mc G$.

    \begin{lemma}
        \label{lem:dag}
        $\Phi(\mc G, R_i) \eq \Phi(\mc G, P, L, i) \eq \{ PLw \; : \; w \in \mbb R^i, \s[8] w_j = 0 \s[8] \tu{for all} \s[8] j > |\pa{\mc G}{i + 1}| \}$
    \end{lemma}

    \begin{proof}
        Recall $P$ is a permutation matrix where $P^\top R_i P$ permutes the rows and columns of $R_i$ such that rows and columns corresponding to parents of $i+1$ come first. Accordingly, $b_{i+1}^\top P$ permutes the columns of $b_{i+1}^\top$ such that columns corresponding to parents of $i+1$ come first and the following are equivalent:
        \begin{itemize}
            \item $(b_{i+1}^\top)_j = 0 \s[8] \tu{for all} \s[8] j \not \in \pa{\mc G}{i+1}$;
            \item $(b_{i+1}^\top P)_j = 0 \s[8] \tu{for all} \s[8] j > |\pa{\mc G}{i + 1}|$.
        \end{itemize}
      


        \noindent By Equation \ref{eq:beta_conv}:
        \[
            b_{i+1}^\top P = w^\top L^{-1}
        \]
        where the $j^\text{th}$ element of $b_{i+1}^\top P$ is computed:
        \[
            (b_{i+1}^\top P)_{j} = \sum_{k = 1}^i (w^\top)_{k} \, (L^{-1})_{k,\,j}
        \]
        Moreover, $L^{-1}$ is lower triangular so $(L^{-1})_{k,\,j} = 0$ for all $k < i$:
        \[
            (b_{i+1}^\top P)_{j} = \sum_{k = j}^i (w^\top)_{k} \, (L^{-1})_{k,\,j}
        \]

        \noindent We show $(b_{i+1}^\top P)_{j} = 0$ for all $j > \pa{\mc G}{i + 1}$ if and only if $(w^\top)_{j} = 0$ for all $j > \pa{\mc G}{i + 1}$ by induction:
        \begin{itemize}
            \item Base case: If $i = j$, then $(b_{i+1}^\top P)_j = (w^\top)_j (L^{-1})_{j,j}$. Recall $L^{-1}$ has real and positive diagonal elements so $(b_{i+1}^\top P)_j = 0$ if and only if $(w^\top)_j = 0$.
            \item Inductive step: If $i > j$ and $(b_{i+1}^\top P)_{k} = (w^\top)_{k} = 0$ for all $k > j$, then:
            \[
                (b_{i+1}^\top P)_{j} = \sum_{k = j}^i (w^\top)_{k} \, (L^{-1})_{k,\,j} = (w^\top)_j (L^{-1})_{j,j}
            \]
            Recall $L^{-1}$ has real and positive diagonal elements so $(b_{i+1}^\top P)_{j} = 0$ if and only if $(w^\top)_{j} = 0$.
        \end{itemize}
            
        \noindent Accordingly, the following are equivalent:
        \begin{itemize}
            \item $(b_{i+1}^\top)_j = 0 \s[8] \tu{for all} \s[8] j \not \in \pa{\mc G}{i+1}$;
            \item $(w^\top)_{j} = 0 \s[8] \tu{for all} \s[8] j > |\pa{\mc G}{i + 1}|$.
        \end{itemize}
    Making the appropriate substitution in $\Phi(\mc G, R_i)$ proves the lemma.
    \end{proof}

    \begin{lemma}
        \label{lem:corr}
        $\Psi(R_i) \eq \Psi(P, L, i) \eq \{ PLw \; : \; w \in \mbb R^i, \s[8] w^\top w < 1 \}$
    \end{lemma}

    \begin{proof}
        By construction $r = PLw$ where $PL$ is a bijection by Lemma \ref{lem:biject}. Furthermore, $r_{i+1}^\top R_i^{-1} r_{i+1} = w^\top w$ by Equation \ref{eq:omega_conv}. Then, making the appropriate substitution in $\Psi(R_i)$ proves the lemma.
    \end{proof}

    Theorem \ref{thm:dao} proves that the DaO method samples positive definite correlation matrices that are Markov to $\mc G$ uniformly at random.

    \begin{theorem}
        \label{thm:dao}
        If $\mc G = (V, E)$ is a DAG and $m = | \{ i \in V \; : \; \pa{\mc G}{i} = \es \} |$, then initializing $R_m = I_{m \times m}$ and $\gamma_m = \frac{p-m}{2}$ and iteratively sampling completions from $\Phi(\mc G, P, L, i) \cap \Psi(P, L, i)$ where $w = \begin{bmatrix} w_1 \\ w_2 \end{bmatrix}$ where $w_1 \sim \textit{mPII}_k(\gamma_i)$ and $\gamma_{i+1} = \gamma_i - \frac{1}{2}$ samples $R_p$ uniformly from the space of correlation matrices Markov with respect to $\mc G$.
    \end{theorem}

    \begin{proof}
        By Lemmas \ref{lem:dag} and \ref{lem:corr}, this procedure samples the space of correlation matrices Markov with respect to $\mc G$.

        \vskip 5mm

        \noindent Noting $f(R_m \mid \gamma_m) = \det(R_m)^{\gamma_m} = 1$ and the following relation derived from Schur's compliment \citep{ouellette1981schur}:
        \begin{align*}
            \det(R_{i+1}) &= \det(R_{i}) \det(1 - r_{i+1}^\top R_{i}^{-1} r_{i+1}) \\
            &= (1 - r_{i+1}^\top R_{i}^{-1} r_{i+1}) \det(R_{i})
        \end{align*}
        we get the following recurrence relation for the completion $r_{i+1}$ of $R_i$ in $R_{i+1}$:
        \begin{align*}
            f(r_{i+1} \mid R_{i}, \, \gamma_i) &\propto (1 - r_{i+1}^\top R_i^{-1} r_{i+1})^{\gamma_i - \frac{1}{2}} \; \det(R_i)^{-\frac{1}{2}} \\
            &= (1 - r_{i+1}^\top R_{i}^{-1} r_{i+1})^{\gamma_i - \frac{1}{2}} \; \det(R_{i})^{-\frac{1}{2}} \; \frac{\det(R_{i})^{\gamma_i}}
            {\det(R_{i})^{\gamma_i}} \\
            &= \frac{(1 - r_{i+1}^\top R_{i}^{-1} r_{i+1})^{\gamma_i - \frac{1}{2}} \; \det(R_{i})^{\gamma_i - \frac{1}{2}}}{\det(R_{i})^{\gamma_i}} \\
            &= \frac{\left[(1 - r_{i+1}^\top R_{i}^{-1} r_{i+1}) \; \det(R_{i})\right]^{\gamma_i - \frac{1}{2}}}{\det(R_{i})^{\gamma_i}} \\
            &= \frac{\det(R_{i+1})^{\gamma_{i+1}}}{\det(R_{i})^{\gamma_i}}
        \end{align*}
    
        \noindent Lastly, by noting $f(R_{i+1} \mid \gamma_{i+1}) = f(r_{i+1} \mid R_i, \, \gamma_i) \; f(R_i \mid \gamma_i)$:
        \begin{align*}
            f(R_{p} \mid \gamma_{p}) &= \prod_{i=m}^{p-1} f(r_{i+1} \mid R_{i}, \, \gamma_i) \; f(R_m \mid \gamma_m) \\
            &\propto \prod_{i=m}^{p-1} \frac{\det(R_{i+1})^{\gamma_{i+1}}}{\det(R_{i})^{\gamma_i}} \; \det(R_m)^{\gamma_m} \\
            &= \det(R_{p})^{\gamma_p}
        \end{align*}
        where letting $\gamma_m = \frac{p-m}{2}$ implies $\gamma_i = \frac{p-i}{2}$ and $\gamma_p = 0$ which samples $R_p$ uniformly.
    \end{proof}

\subsection{The DAG-adaptation of the Onion Method}
\label{sec:dao}

    Algorithm \ref{alg:dao} ($\mt{DaO}$) provides pseudocode for the DaO method and employs two subroutines: 
    \begin{itemize}
        \item Algorithm \ref{alg:mpii} ($\mt{mPII}$) returns a $d$-dimensional column vector where the first $k$ entries are drawn from a multivariate Pearson type II distribution and the last $d-k$ entries are set to zero.
        \item Algorithm \ref{alg:pmat} ($\mt{P\tu{-}mat}$) returns a permutation matrix $P$ where $P^\top R_i P$ permutes the rows and columns of $R_i$ such that rows and columns corresponding to parents of $i+1$ come first.
    \end{itemize}
    In Algorithm \ref{alg:mpii}$ (\mt{mPII}$): $\mt{beta}(a, b)$ samples a beta distribution with parameters $a$ and $b$; $\mt{randn}(k)$ samples a $k$-dimensional standard normal distribution; and Lines 3--4: uniformly sample the surface of the unit $k$-sphere using a procedure outlined by \cite{muller1956some}.

    \noindent
    \begin{minipage}{0.58\textwidth} \begin{algorithm}[H]
    \DontPrintSemicolon
    \setstretch{1.1}
    \caption{$\mt{mPII}(\mc G, i)$}
    \label{alg:mpii}
    \KwIn{$\mt{DAG}: \mc G = (V, E) \; \text{s.t.} \; |V| = p \s[18] \mt{idx}: i$}
    \KwOut{$\mt{vector}: w$}
    $k \at |\pa{\mc G}{i + 1}| \; ; \s[24] \gamma \at \frac{p - i}{2}$ \;
    $q \sim \mt{beta}(\frac{k}{2}, \; \gamma + \frac{1}{2})$ \s[82] \tcp{$q \neq 1$}
    $y \sim \mt{randn}(k)$ \s[128] \tcp{$\Vert y \Vert_2 \neq 0$}
    $u \at \dfrac{y}{\vphantom{\big|} \s[5] \Vert y \Vert_2}$ \;
    $w \at 
    \begin{bmatrix}
         q^\frac{1}{2} u \\
         0_{(i-k)}
    \end{bmatrix}$ \;
\end{algorithm} \end{minipage}
    \begin{minipage}{0.01\textwidth} \hfill \end{minipage}
    \begin{minipage}{0.4\textwidth} \begin{algorithm}[H]
    \DontPrintSemicolon
    \setstretch{1.1}
    \caption{$\mt{P\tu{-}mat}(\mc G, i)$}
    \label{alg:pmat}
    \KwIn{$\mt{DAG}: \mc G \s[18] \mt{idx}: i$}
    \KwOut{$\mt{matrix}: P$}
    $P \at 0_{i \times i} \; ; \s[24] k \at 1$ \;
    \For{$j \at 1$ \KwTo $i$}{
        \If{$j \in \pa{\mc G}{i+1}$}{
            $P_{j,k} \at 1 \; ; \s[12] k \at k + 1$ \;
        }
    }
    \For{$j \at 1$ \KwTo $i$}{
        \If{$j \not \in \pa{\mc G}{i+1}$}{
            $P_{j,k} \at 1 \; ; \s[12] k \at k + 1$ \;
        }
    }
\end{algorithm}

 \end{minipage}
    \begin{algorithm}[H]
    \DontPrintSemicolon
    \setstretch{1.1}
    \caption{$\mt{DaO}(\mc G)$}
    \label{alg:dao}
    \KwIn{$\mt{DAG}: \mc G = (V, E) \; \text{s.t.} \; |V| = p$}
    \KwOut{$\mt{corr}: R \s[24] \mt{coef}: B \s[24] \mt{err}\tu{-}\mt{var}: \Omega$}
    $m = | \{ i \in V \; : \; \pa{\mc G}{i} = \es \} |$ \;
    $R_m \at I_{m \times m} \; ; \s[24] B_m \at 0_{m \times m} \; ; \s[24] \Omega_m \at I_{m \times m}$ \;
    \For{$i \at m$ \KwTo $p-1$}{
        $w \sim \mt{mPII}(\mc G, i)$ \;
        $P \at \mt{P\tu{-}mat}(\mc G, i)$ \;
        $L L^\top \at P^\top R_i P$ \s[82] \tcp{Cholesky decomposition}
        $r_{i+1} \at P L w \; ; \s[90] b_{i+1}^\top \at w^\top L^{-1} P^\top ; \s[32] \omega_{i+1} \at 1 - w^\top w$ \;
        $R_{i+1} \at 
        \begin{bmatrix}
            R_{i} & r_{i+1} \\
            r_{i+1}^\top & 1
        \end{bmatrix}; \s[24]
        B_{i+1} \at 
        \begin{bmatrix}
            B_{i} & 0_i \\
            b_{i+1}^\top & 0
        \end{bmatrix}; \s[24]
        \Omega_{i+1} \at 
        \begin{bmatrix}
            \Omega_{i} & 0_i \\
            0_i^\top & \omega_{i+1}
        \end{bmatrix}$ \;
    }
\end{algorithm}
    In Algorithms \ref{alg:mpii}, \ref{alg:pmat}, and \ref{alg:dao}: $I_{m \times m}$ and $0_{m \times m}$ denote the $m \times m$ identity matrix and zero matrix respectively while $0_{(i-k)}$ denotes a column vector of $i-k$ zeros. 
    
    \begin{remark}
        \cite{lewandowski2009generating} give an extension to the onion method that samples correlation matrices proportional to a power of their determinant by choosing a different initial value for the parameter that we call $\gamma_i$. The same extension could be done here, but we did not add this parameter to the DaO method because we prioritized minimizing the number of arguments.
    \end{remark}

    Figure \ref{fig:dao} depicts the results of generating random correlation matrices in the simplest non-trivial cases using the DaO method. The cases are depicted in the first column of Table \ref{tab:er_dags}---from top to bottom, the rows correspond to the emitter, chain, and collider cases, respectively.

    \begin{figure}
        \centering
        \includegraphics[scale=0.5]{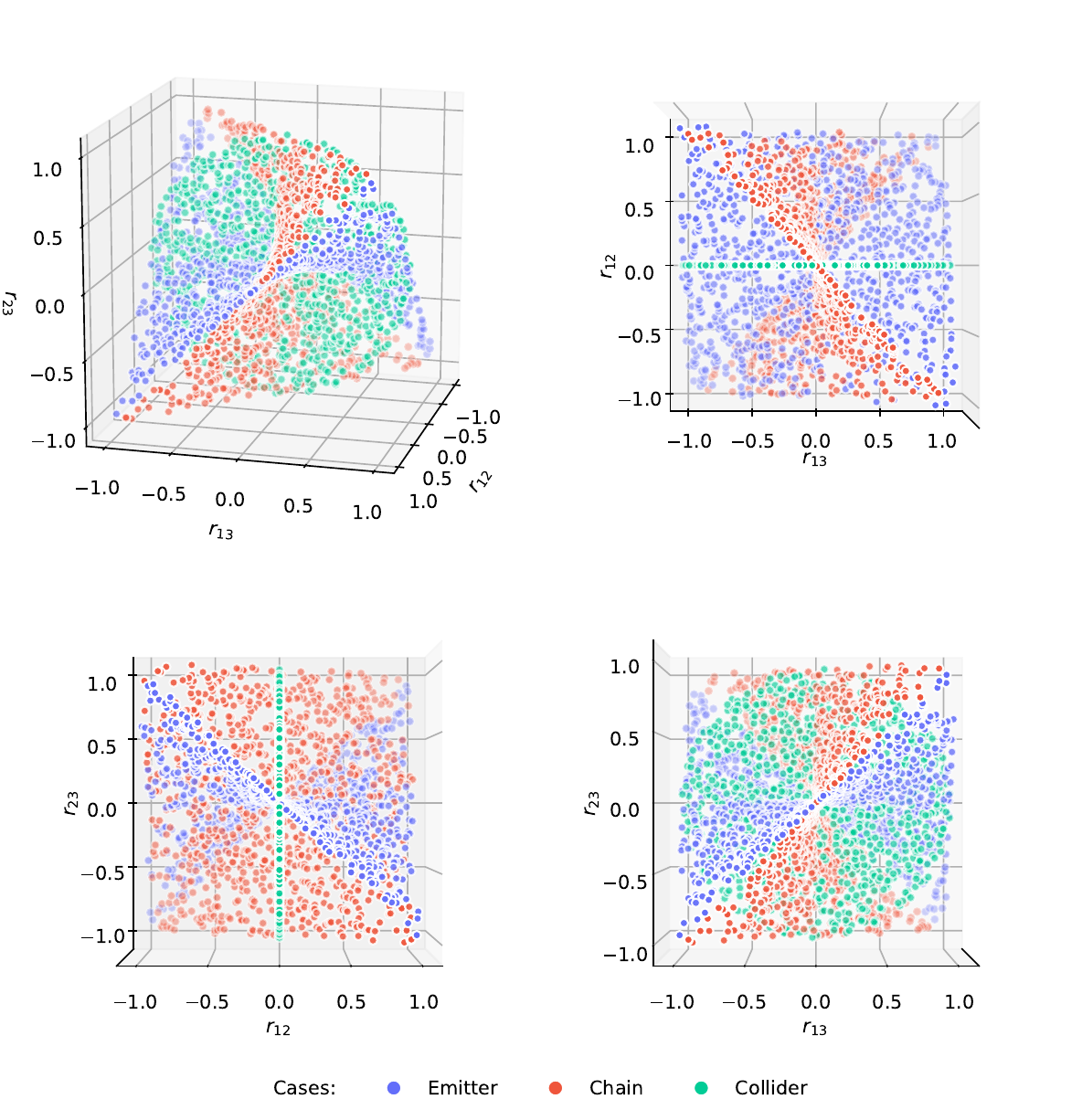}
        \caption{Uniformly sampled correlation matrices from DAGs with $|V| = 3$ and $|E| = 2$ corresponding to the $1 < 2 < 3$ column of Table \ref{tab:er_dags} with 100 repetition for each case. These are 2D projections of a 3D space, so many points are occluded by other points closer to the viewer.}
        \label{fig:dao}
    \end{figure}
    
\section{Evaluation of Parametric Models}
\label{sec:eval_of_params}

    This section evaluates and compares parametric models produced by the DaO method against parametric models produced by two alternative simulation designs. We evaluate and compare these methods on three different DAG types: Erd\H{o}s R{\'e}nyi (ER), scale-free in-degree (SFi), and scale-free out-degree (SFo). We ``standardize'' the models produced by the alternative methods to facilitate a meaningful comparison. This ensures all variables have unit variance and that we compare scale independent versions of the models.
    
    For convenience, we restate Equation \ref{eq:sigma}:
    \[
        \Sigma = (I - B)^{-\top} \, \Omega \, (I - B)^{-1}.
    \]
    The two alternative simulation designs parameterize $B$ and $\Omega$ directly where $B$ is a matrix of beta coefficients and $\Omega$ is a covariance matrix for additive error terms.

    \begin{itemize}
        \item Zheng, Aragam, Ravikuma, and Xing (ZARX) simulation \citep{zheng2018dags}:\footnote{This simulation was originally used by \cite{zheng2018dags} but nearly all continuous optimization-based CDA publications use this design as well.}
        \begin{align*}
            B &= (\beta_{i,j}) \; : \; 
            \begin{cases}
                \beta_{i,j} \sim \textit{uniform}( \, [ \, -2.0, \, -0.5 \, ] \, \cup \, [ \, 0.5, \, 2.0 \, ] \, ) & \s[8] \text{if} \s[8] j \in \pa{\mc G}{i} \\
                \beta_{i,j} = 0 & \s[8] \text{if} \s[8] j \not \in \pa{\mc G}{i}
            \end{cases} \\[2mm]
            \Omega &= (\omega_{i,j}) \; : \; 
            \begin{cases}
                \omega_{i,j} = 1 & \s[8] \text{if} \s[8] i = j \\
                \omega_{i,j} = 0 & \s[8] \text{if} \s[8] i \neq j
            \end{cases}
        \end{align*}
        \item Tetrad simulation \citep{andrews2023fast}:\footnote{Tetrad is a Java application for causal discovery. Simulations in Tetrad are flexible, however, this design was used in a recent publication whose authors included Joseph Ramsey---the lead developer and maintainer of Tetrad.}
        \begin{align*}
            B &= (\beta_{i,j}) \; : \; 
            \begin{cases}
                \beta_{i,j} \sim \textit{uniform}( \, [ \, -1.0, \, 1.0 \, ] \, ) & \s[8] \text{if} \s[8] j \in \pa{\mc G}{i} \\
                \beta_{i,j} = 0 & \s[8] \text{if} \s[8] j \not \in \pa{\mc G}{i}
            \end{cases} \\[2mm]
            \Omega &= (\omega_{i,j}) \; : \; 
            \begin{cases}
                \omega_{i,j} \sim \textit{uniform}( \, [ \, 1.0, \, 2.0 \, ] \, ) & \s[8] \text{if} \s[8] i = j \\
                \omega_{i,j} = 0 & \s[8] \text{if} \s[8] i \neq j
            \end{cases}
        \end{align*}
    \end{itemize}

    The methods above are ``standardized'' using Algorithms \ref{alg:cov_corr} and \ref{alg:cov_dag}. Algorithm \ref{alg:cov_corr} ($\mt{Cov}$-$\mt{to}$-$\mt{Corr}$) converts a covariance matrix to a correlation matrix and Algorithm \ref{alg:cov_dag} ($\mt{Cov}$-$\mt{to}$-$\mt{DAG}$) takes a DAG and covariance/correlation matrix as input and returns the corresponding $B$ and $\Omega$ parameters for the model---the parameters are standardized if the input is a DAG and a correlation matrix.
    
    \noindent
    \begin{minipage}{0.48\textwidth} \begin{algorithm}[H]
    \DontPrintSemicolon
    \setstretch{1.1}
    \caption{$\mt{Cov\tu{-}to\tu{-}Corr}(\Sigma)$}
    \label{alg:cov_corr}
    \KwIn{$\mt{cov}: \Sigma_{(p \times p)}$}
    \KwOut{$\mt{corr}: R$}
    $D \at 0_{p \times p}$ \;
    \For{$i \at 1$ \KwTo $p$}{
        $D_{i,i} \at (\Sigma_{i,i})^{-\frac{1}{2}}$ \;
    }
    $R \at D \, \Sigma \, D$ \; 
\end{algorithm} \end{minipage}
    \begin{minipage}{0.01\textwidth} \hfill \end{minipage}
    \begin{minipage}{0.48\textwidth} 

\begin{algorithm}[H]
    \DontPrintSemicolon
    \setstretch{1.1}
    \caption{$\mt{Cov\tu{-}to\tu{-}DAG}(\mc G, \Sigma)$}
    \label{alg:cov_dag}
    \KwIn{$\mt{DAG}: \mc G \s[18] \mt{cov}/\mt{corr}: \Sigma_{(p \times p)}$}
    \KwOut{$\mt{coef}: B \s[18] \mt{err}\tu{-}\mt{var}: \Omega$}
    $B \at 0_{p \times p} \; ; \s[24] \Omega \at 0_{p \times p}$ \;
    \For{$i \at 1$ \KwTo $p$}{
        $J \at \pa{\mc G}{i}$ \;
        $B_{i,J} \at \Sigma_{i,J} \, (\Sigma_{J,J})^{-1}$ \;
        $\Omega_{i,i} \at \Sigma_{i,i} - B_{i,J} \, \Sigma_{J,i}$ \;
    }
\end{algorithm} \end{minipage}

    In Algorithms \ref{alg:cov_corr} and \ref{alg:cov_dag}: $0_{p \times p}$ denotes the $p \times p$ zero matrix while the subscript of $\Sigma_{(p \times p)}$ denote its dimensions. We first evaluate the distributions of edge coefficients, implied absolute correlations, and the standard deviations of the independent error terms for each of these simulation methods. We then evaluate the $R^2$-sortability of these methods.

\subsection{Comparison of Simulation Designs}
\label{sec:compare_params}

    In this section, we compare the DaO method against the ZARX and Tetrad methods. Figure \ref{fig:beta_scatter} plots the magnitude of the standardized betas against the standard deviation of the additive error term for every edge with kernel density estimates for contours. Figures \ref{fig:er_sims}, \ref{fig:sfi_sims}, and \ref{fig:sfo_sims} depict empirical cumulative distribution functions (ECDFs) of the absolute beta coefficients, absolute correlations, and standard deviations of 100 models generated from DAGs with 100 variables and average degree $\alpha=10$. The values where the corresponding implied correlation is zero or one are omitted and vertical dotted lines indicate where the maximum value of each ECDF is obtained. Lastly, \cite{zheng2018dags} recommend using $0.3$ as a threshold for the absolute beta coefficients to determine if a non-zero value corresponds to an edge in continuous optimization-based methods---we mark this threshold with a vertical dashed line.

    \begin{figure}[ht]
        \centering
        \includegraphics[scale=0.6]{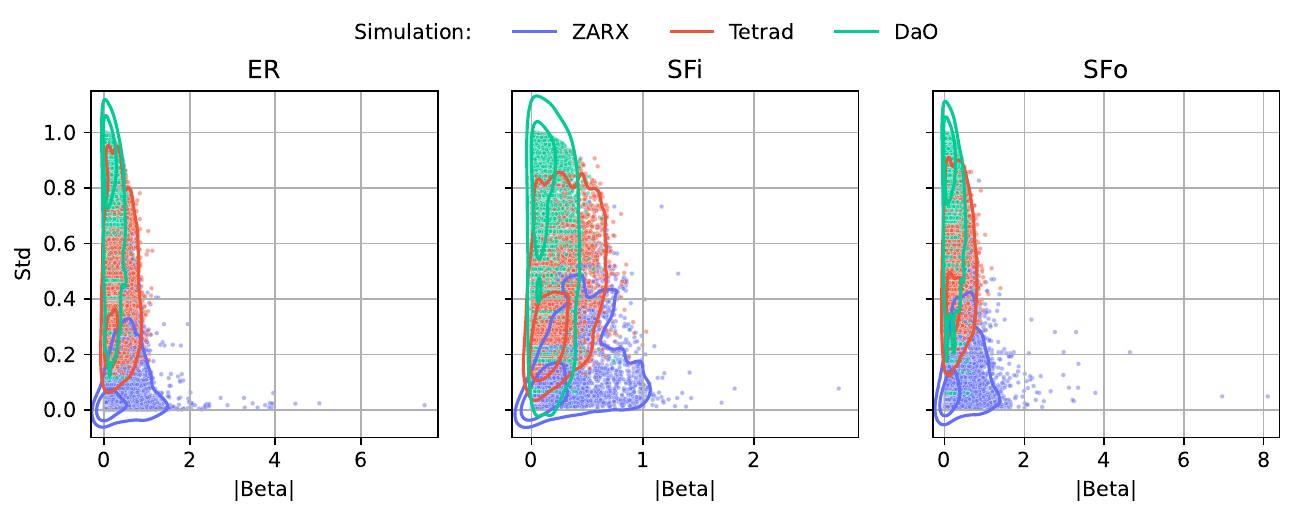}
        \caption{Properties of 10 models generated from ER/SFi/SFo-DAGs with $|V| = 100$ and $\alpha = 10$.}
        \label{fig:beta_scatter}
    \end{figure}

    \begin{figure}[ht]
        \centering
        \includegraphics[scale=0.6]{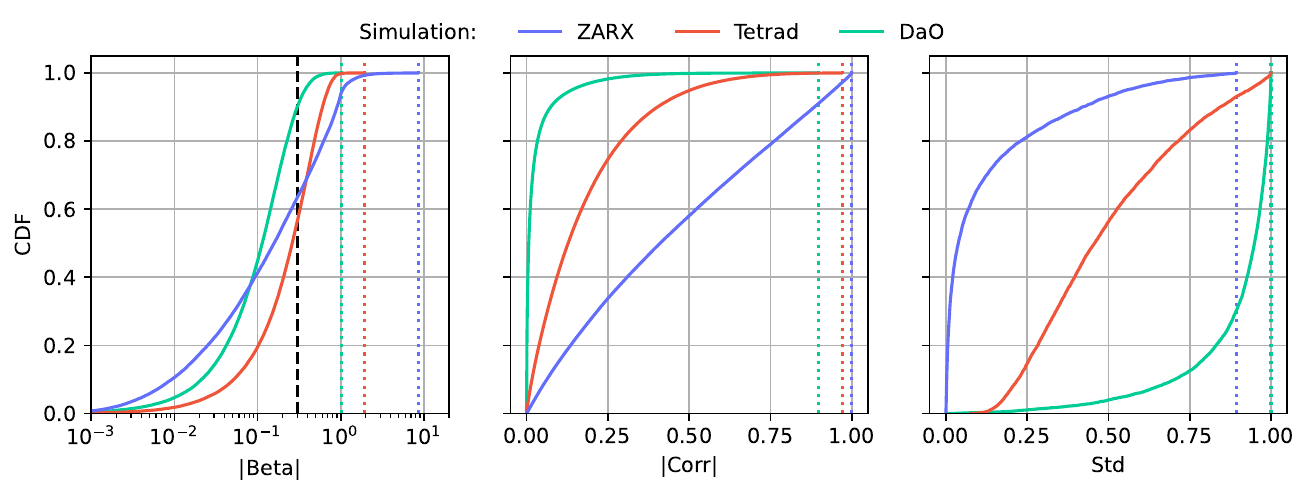}
        \caption{Properties of 100 models generated from ER-DAGs with $|V| = 100$ and $\alpha = 10$.}
        \label{fig:er_sims}
    \end{figure}
    
    \begin{figure}[ht]
        \centering
        \includegraphics[scale=0.6]{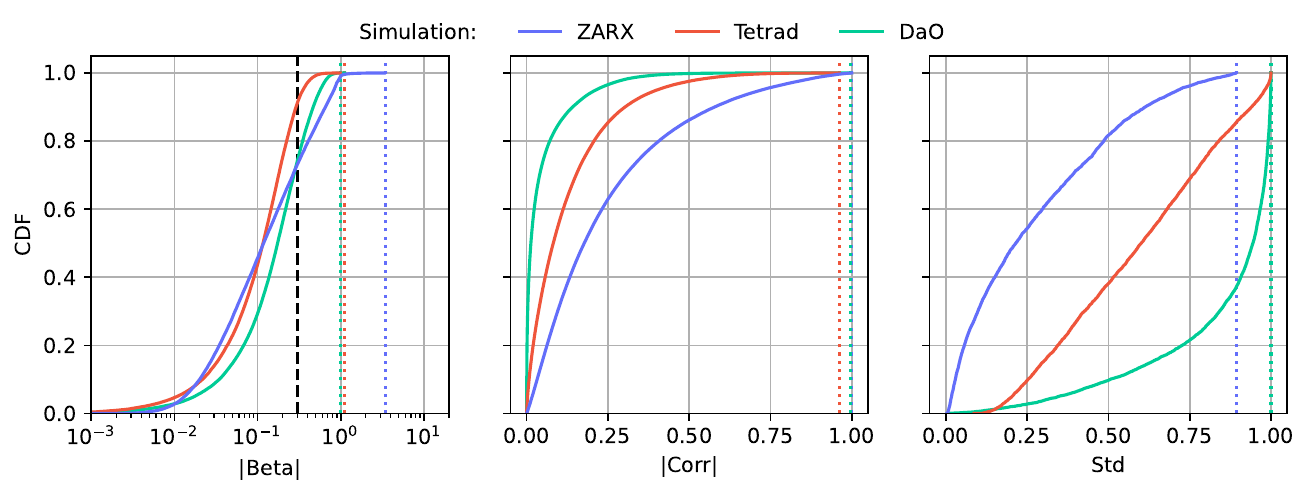}
        \caption{Properties of 100 models generated from SFi-DAGs with $|V| = 100$ and $\alpha = 10$.}
        \label{fig:sfi_sims}
    \end{figure}

    \begin{figure}[ht]
        \centering
        \includegraphics[scale=0.6]{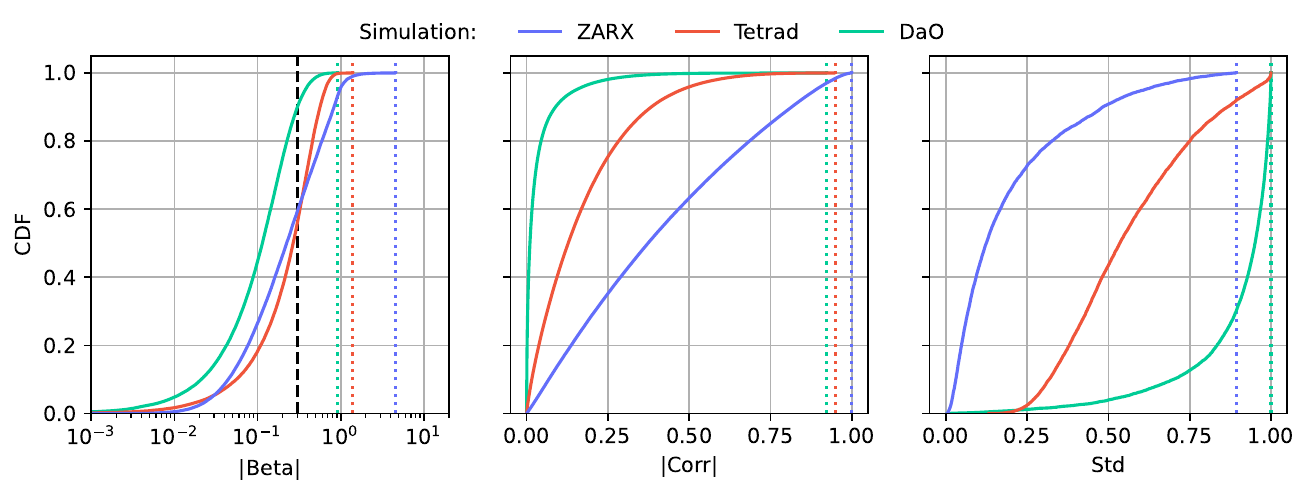}
        \caption{Properties of 100 models generated from SFo-DAGs with $|V| = 100$ and $\alpha = 10$.}
        \label{fig:sfo_sims}
    \end{figure}

    All three simulation methods generate substantially different distributions of parametric models. The ZARX method consistently produces more edges with large absolute beta coefficients relative to the other two methods, while DaO produces edges with absolute beta coefficients closer to zero. Similarly, the ZARX method produces the highest absolute correlation values, while DaO produces the lowest absolute correlation values. Correspondingly, the standard deviations of the error terms in ZARX models are the smallest, while the standard deviations of error terms in DaO models are the largest. 
    
    The vertical dashed lines in these figures indicate that the standard threshold used by continuous optimization-based methods rules out over 50\% of the true edges in the standardized versions of all three simulations. As such, all edges to the left of the vertical dashed lines can not be discovered by continuous optimization-based methods unless the user changes the threshold value---in Section \ref{sec:sim_study} we lower the threshold to $0.1$ (setting the threshold lower often results in cyclic models).

    The ZARX method samples beta coefficients from a distribution bounded away from zero in an attempt to avoid (near-)violations of faithfulness and to generate models with a strong signal-to-noise ratio. Our experiments show that this is a misconception. After standardization, the beta coefficients produced by ZARX simulations are no longer bounded away from zero and there are many weak effects. In fact, the beta coefficients produced by Tetrad simulations tend to be further from zero than the beta coefficients produced by ZARX simulations. Moreover, ZARX simulations produce many additive error terms with a standard deviation of approximately zero---this near determinism will induce (near-)violations of faithfulness.

    Also puzzling is the presence of beta coefficients with magnitude approximately equal to 10 in the standardized ZARX models. This is only possible if two or more parents vertices are almost perfectly (anti-)correlated and together have a jointly cancelling effect on a shared child vertex.

\subsection{$R^2$-Sortability of Parametric Models}
\label{sec:r2sort_of_params}

    Unlike varsortability which can be prevented by standardization \citep{reisach2021beware}, it is more difficult to prevent $R^2$-sortability \citep{reisach2023simple}. In this section, we evaluate the $R^2$-sortability of the parametric models produced by the three simulation methods. We generated 100 models with 20 variables and 100 models with 100 variables for each model structure: ER, SFi, and SFo. All models were generated with average degree $\alpha = 10$. 
    
    We quantified the $R^2$-sortability of each simulation method on each model structure type by calculating the rank correlation coefficient between each variable's $R^2$-sortability rank and its index from the variable order used to generate the DAG; see Section \ref{sec:sampling_dags}. 
    
    Simulated $R^2$-sortable data will have correlations approaching plus or minus one, while, simulated non-$R^2$-sortable data will have correlations closer to zero. Tables \ref{tab:r2s_20} and \ref{tab:r2s_100} show these correlations for the 20 variable and 100 variable models, respectively. Interestingly, we found that $R^2$-sortability varies substantially with model structure, as ER models typically had the highest $R^2$-sortability correlations, and SFo models typically had the lowest $R^2$-sortability correlations.

    \noindent
    \begin{minipage}{0.48\textwidth} \begin{table}[H]
    \centering
    \begin{tabular}{cccc}
        \toprule
        & ZARX & Tetrad & DaO \\
        \cmidrule(lr){2-2}
        \cmidrule(lr){3-3}
        \cmidrule(lr){4-4}
        ER & -0.859 & -0.500 & -0.264 \\
        SFi & -0.788 & -0.532 & -0.223 \\
        SFo & -0.624 & -0.126 & 0.033 \\
        \bottomrule
    \end{tabular}
    \caption{$R^2$ rank correlation on DAGs with $|V|=20$ and $\alpha = 10$ over 100 repetitions}
    \label{tab:r2s_20}
\end{table} \end{minipage}
    \begin{minipage}{0.01\textwidth} \hfill \end{minipage}
    \begin{minipage}{0.48\textwidth} \begin{table}[H]
    \centering
    \begin{tabular}{cccc}
        \toprule
        & ZARX & Tetrad & DaO \\
        \cmidrule(lr){2-2}
        \cmidrule(lr){3-3}
        \cmidrule(lr){4-4}
        ER & -0.873 & -0.557 & -0.386 \\
        SFi & -0.552 & -0.336 & -0.108 \\
        SFo & -0.497 & -0.051 & -0.084 \\
        \bottomrule
    \end{tabular}
    \caption{$R^2$ rank correlation on DAGs with $|V|=100$ and $\alpha = 10$ over 100 repetitions}
    \label{tab:r2s_100}
\end{table} \end{minipage}
    \vskip 5mm


    
    
    

    \begin{figure}[ht]
        \centering
        \includegraphics[scale=0.6]{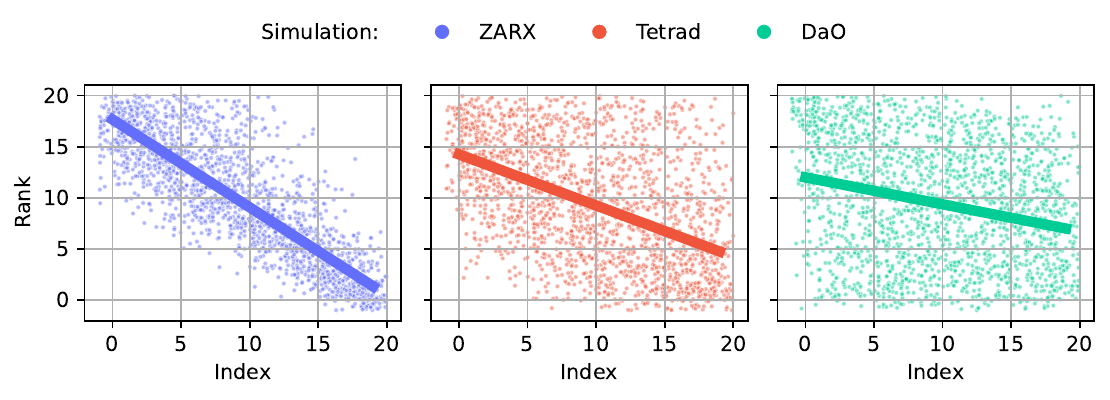}
        \caption{$R^2$-sortability of models generated from ER-DAGs with $|V| = 20$ and $\alpha = 10$ over 100 repetitions.}
        \label{fig:er_r2s_20}
    \end{figure}

    \begin{figure}[ht]
        \centering
        \includegraphics[scale=0.6]{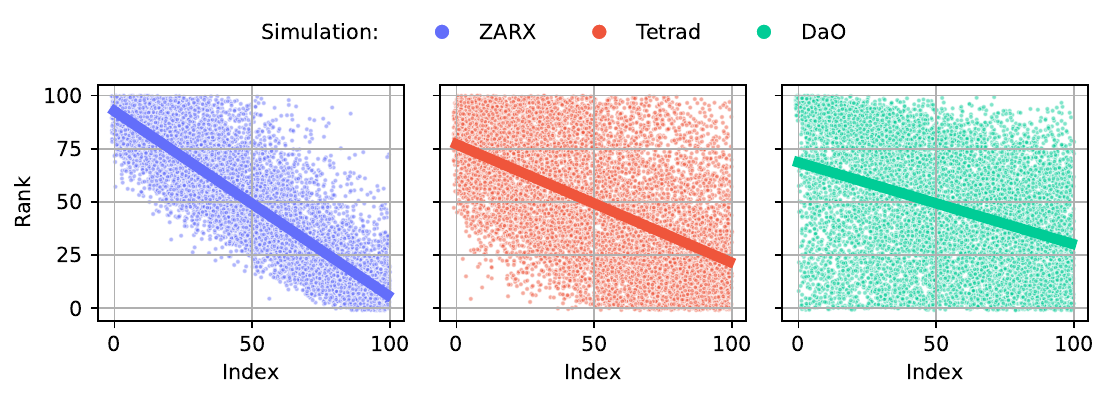}
        \caption{$R^2$-sortability of models generated from ER-DAGs with $|V| = 100$ and $\alpha = 10$ over 100 repetitions.}
        \label{fig:er_r2s_100}
    \end{figure}

    Tables \ref{tab:r2s_20} and \ref{tab:r2s_100} tabulate the levels of $R^2$-sortability for the three simulation methods while varying model structure. ZARX simulations have the highest $R^2$-sortability for all model structures, while $R^2$-sortability appears to be present in Tetrad simulations but at a weaker level than in ZARX simulations. Overall, DaO has lower $R^2$-sortability levels.

    

    \begin{figure}[ht]
        \centering
        \includegraphics[scale=0.6]{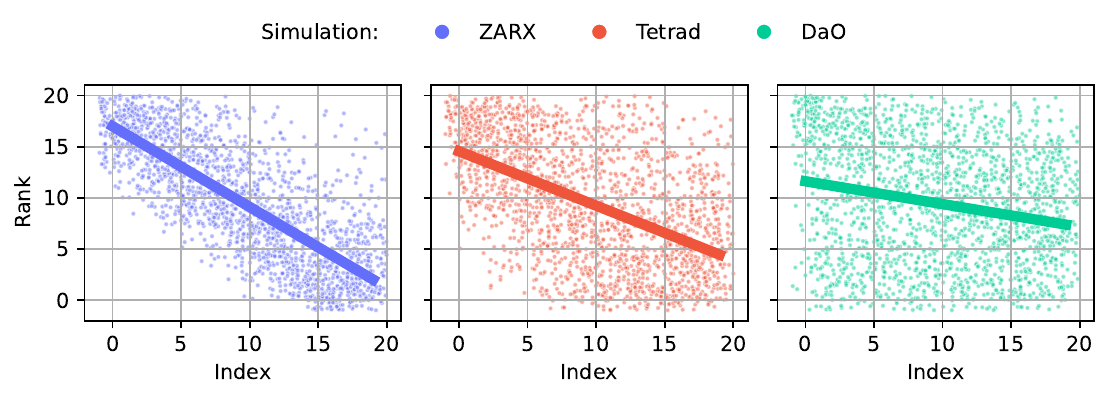}
        \caption{$R^2$-sortability of models generated from SFi-DAGs with $|V| = 20$ and $\alpha = 10$ over 100 repetitions.}
        \label{fig:sfi_r2s_20}
    \end{figure}

    \begin{figure}[ht]
        \centering
        \includegraphics[scale=0.6]{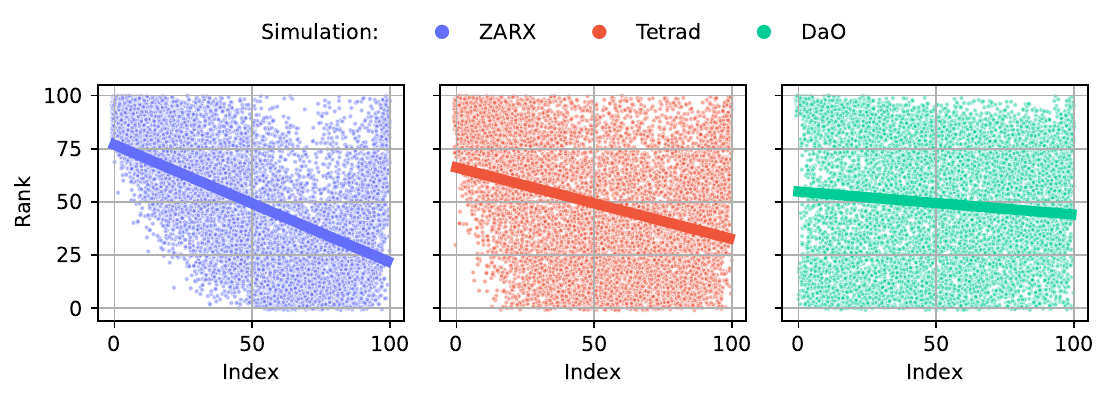}
        \caption{$R^2$-sortability of models generated from SFi-DAGs with $|V| = 100$ and $\alpha = 10$ over 100 repetitions.}
        \label{fig:sfi_r2s_100}
    \end{figure}

    \begin{figure}[ht]
        \centering
        \includegraphics[scale=0.6]{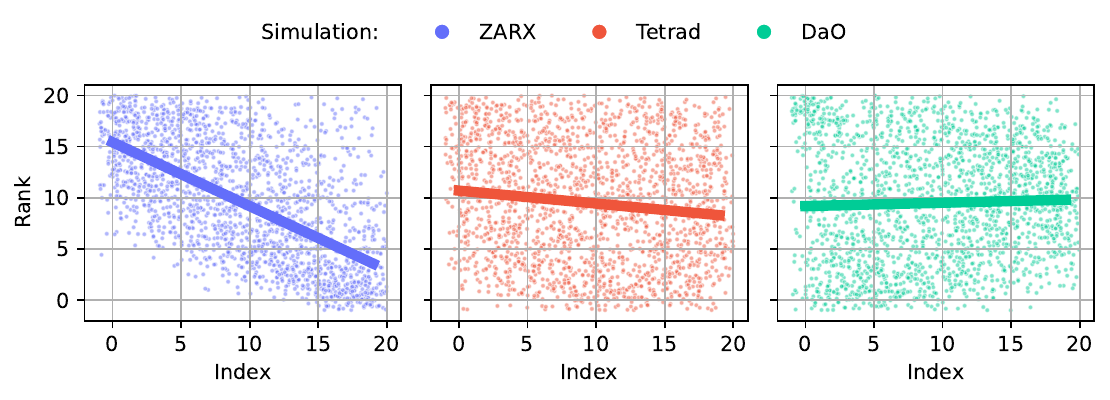}
        \caption{$R^2$-sortability of models generated from SFo-DAGs with $|V| = 20$ and $\alpha = 10$ over 100 repetitions.}
        \label{fig:sfo_r2s_20}
    \end{figure}
    
    \begin{figure}[ht]
        \centering
        \includegraphics[scale=0.6]{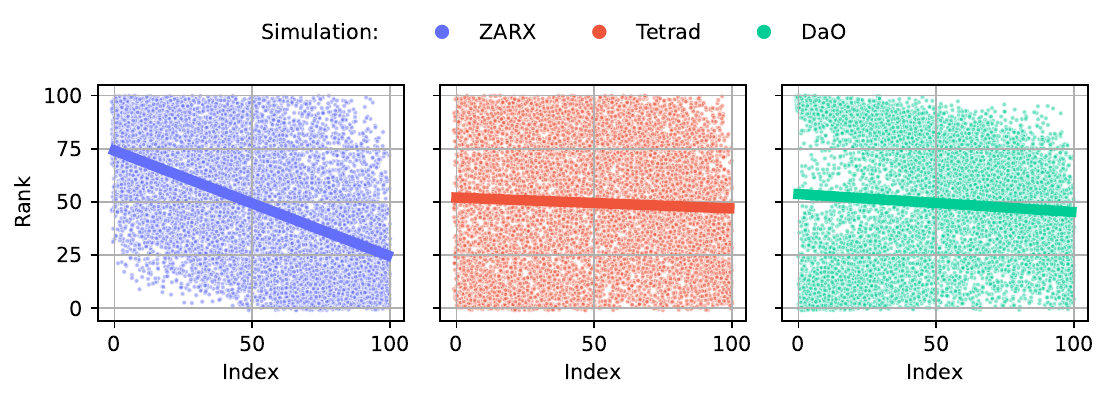}
        \caption{$R^2$-sortability of models generated from SFo-DAGs with $|V| = 100$ and $\alpha = 10$ over 100 repetitions.}
        \label{fig:sfo_r2s_100}
    \end{figure}

\section{Simulation Study}
\label{sec:sim_study}

    In this section, we report the results of a small simulation study testing a variety of CDAs on data generated using the DaO method and two alternatives: the ZARX and Tetrad methods.

    \begin{algorithm}
    \DontPrintSemicolon
    \setstretch{1.1}
    \caption{$\mt{simulate}(B, \Omega, \mc E, n)$}
    \label{alg:sim}
    \KwIn{$\mt{coef}: B_{(p \times p)} \s[18] \mt{err}\tu{-}\mt{var}: \Omega_{(p \times p)} \s[18] \mt{error}: \mc E \s[18] \mt{samples}: n$}
    \KwOut{$\mt{data}: X$}
    $X \at 0_{n \times p}$ \;
    \For{$i \at 1$ \KwTo $p$}{
        $Z_{i} \sim \mc E(\Omega_{i,i}, n)$ \;
        $X_{*,i} \at X \, (B_{i,*})^\top + Z_{i}$ \;
    }
\end{algorithm}
    In Algorithm \ref{alg:sim}: $0_{n \times p}$ denotes the $n \times p$ zero matrix while the subscripts of $B_{(p \times p)}$ and $\Omega_{(p \times p)}$ denote their dimensions. The second parameter of $\mc E$ gives the number of samples to draw.

    We used the metrics tabulated in Table \ref{tab:metrics} which were originally proposed by \cite{kummerfeld2023power}. All simulations were run with average degree $\alpha = 10$ over 20 repetitions at sample sizes: 200, 400, 600, 800, 1,000, 1,500, 2,000, 3,000, and 4,000. The number of variables were varied ($|V| = 20$ and $|V| = 100$) in addition to the method of simulation. All algorithms were run on an Apple M1 Pro processor with 16G of RAM. Algorithms were dropped from the analysis in cases where they took longer than 20 minutes per run.
    
    In the simulations, the algorithms we compared use the following tests and scores.
    \begin{itemize}
        \item Fisher Z test \citep{anderson1984introduction, Spirtes1993-od, kalisch2007estimating} is a test of zero-partial correlation using Fisher's z-transform; we set its parameter alpha $= 0.01$.
        \item Bayesian Information Criterion (BIC) score \citep{schwarz1978estimating, haughton1988choice} is a penalized likelihood score that approximates the marginal likelihood for curved exponential families; we set its parameter penalty discount $= 2$.
        \item Bayesian Gaussian equivalent (BGe) score \citep{geiger2002parameter, kuipers2014addendum} is a Bayesian score for a multivariate Gaussian distribution using conjugate priors; we use default parameters.
    \end{itemize}

    The simulations use multiple algorithms from multiple different software packages, outlined below.
    \begin{itemize}
        \item \url{https://github.com/kevinsbello/dagma}
        \begin{itemize}
            \item Directed Acyclic Graphs via M-matrices for Acyclicity (DAGMA) \citep{zheng2018dags, bello2022dagma} is a continuous optimization-based algorithm using an L1 penalty and the Adam optimizer; we set the parameter thresholded $= 0.1$ (or smallest threshold $> 0.1$ that returns a DAG), use defaults for all other parameters, and convert the output to a CPDAG.\footnote{A CPDAG is a graphical summarization of a set of DAGs that imply the same conditional independence relationships. Many CDAs output a CPDAG rather than a DAG.}
        \end{itemize}
        \item \url{https://github.com/cmu-phil/tetrad} \citep{ramsey2018tetrad}
        \begin{itemize}
            \item Best Order Score Search (BOSS) \citep{lam2022greedy, andrews2023fast} searches over permutations which are projected to DAGs using the BIC score; we set the parameter restarts $= 10$ and use defaults for all other parameters. 
            \item Fast Greedy Equivalence Search (fGES) \citep{chickering2002optimal, ramsey2017million} is an optimized version of the greedy equivalent search which searches over equivalence classes of DAGs by iteratively adding then removing edges until the model cannot be improved according to the BIC score; we use default parameters.
            \item Peter Clark (PC) \citep{Spirtes1993-od, colombo2014order} explicitly tests conditional independence relationships to rule out incompatible hypothesis graphs using the Fisher Z test; we set the parameter stable $= \textit{true}$ and use defaults for all other parameters. 
        \end{itemize}
        \item \url{https://github.com/CausalDisco/CausalDisco}
        \begin{itemize}
            \item $\sigma$-Sortability ($\sigma$-Sort) \citep{reisach2021beware} orders variables by iteratively minimizing a variance criterion; we use default parameters.
            \item $R^2$-Sortability ($R^2$-Sort) \citep{reisach2023simple} orders variables by iteratively minimizing a partial correlation criterion; we use default parameters.
        \end{itemize}
        \item \url{https://github.com/cdt15/lingam} \citep{ikeuchi2023python}
        \begin{itemize}
            \item Direct Linear Non-Gaussian Acyclic Model (dLiNGAM) \citep{shimizu2006linear, shimizu2011directlingam} orders the variables by iteratively minimizing a non-Gaussian criterion; we use default parameters.
        \end{itemize}
        \item \url{https://cran.r-project.org/web/packages/BiDAG/index.html} \citep{bidag}
        \begin{itemize}
            \item Iterative Markov Chain Monty Carlo (itMCMC) \citep{kuipers2022efficient} uses MCMC to sample orders that are partially constrained by an adjacency matrix learned by PC using the Fisher Z test and BGe score; we set the parameter hardlimit $= 20$ and use defaults for all other parameters.




            
        \end{itemize}
        \item \url{https://cran.r-project.org/web/packages/pchc/index.html} \citep{pchc}
        \begin{itemize}
            \item Max-Min Hill-Climbing (MMHC) \citep{tsamardinos2006max, tsagris2021new} initially grows Markov blankets by iteratively adding variables that maximize the minimum BIC score over subsets of the current Markov blankets. Explicit tests of conditional independence are then used to rule out incompatible hypothetical graphs using the Fisher Z test; we set the parameter restarts $= 10$ and use defaults for all other parameters.
        \end{itemize}

    \end{itemize}


    
    
    
    \begin{minipage}{0.3\textwidth}
    \vspace{12mm}
    \[
        \text{Precision} = \frac{\mt{tp}}{\mt{tp} + \mt{fp}}
    \]
    \vskip 5mm
    \[
        \text{Recall} = \frac{\mt{tp}}{\mt{tp} + \mt{fn}}
    \]
    \vspace{20mm}
\end{minipage}
\begin{minipage}{0.03\textwidth}
    \hfill
\end{minipage}
\begin{minipage}{0.6\textwidth}
    \begin{table}[H]
        \centering
        \small
        \begin{tabular}{cccc}
            \toprule
            True & Estimated & Adjacency & Orientation \\
            \cmidrule(lr){1-1}
            \cmidrule(lr){2-2}
            \cmidrule(lr){3-3}
            \cmidrule(lr){4-4}
            \multirow{4}{*}{$i \at j$} & $i \at j$ & $\mt{tp}$ & $\mt{tp}, \mt{tn}$ \\
            & $i \ta j$ & $\mt{tp}$ & $\mt{fp}, \mt{fn}$ \\
            & $i - \s[-15] - \s[6] j$ & $\mt{tp}$ & $\mt{fn}$ \\
            & $i \dots j$ & $\mt{fn}$ & $\mt{fn}$ \\
            \cmidrule(lr){1-1}
            \cmidrule(lr){2-2}
            \cmidrule(lr){3-3}
            \cmidrule(lr){4-4}
            \multirow{4}{*}{$i \dots j$} & $i \at j$ & $\mt{fp}$ & $\mt{fp}$ \\
            & $i \ta j$ & $\mt{fp}$ & $\mt{fp}$ \\
            & $i - \s[-15] - \s[6] j$ & $\mt{fp}$ & \\
            & $i \dots j$ & $\mt{tn}$ & \\
            \bottomrule
        \end{tabular}
        \caption{Metrics}
        \label{tab:metrics}
    \end{table}
\end{minipage}

\subsection{DaO Simulations}
\label{sec:dao_sims}

    In the DaO simulations, graphs were generated using the ER-DAG, SFi-DAG, and SFo-DAG methods. The model parameters were used to generate two versions of the simulated data: Gaussian error and non-Gaussian exponential error. The Gaussian version of the simulated data was used to evaluate all CDAs with one exception. The non-Gaussian version of the simulated data was used to evaluate dLiNGAM due to its dependence on a non-Gaussian signal. Accordingly, dLiNGAM had access to additional information that the other algorithms did not, which is reflected in its superior performance. The $\sigma$-sortability algorithm was not evaluated on DaO simulations for a similar reason---there is no signal to sort the variables since all variances are one.

    \begin{figure}
        \centering
        \begin{tabular}{rcc}
            \vphantom{\bigg|} & \s[20] 20 Variables &  \s[20] 100 Variables \\
            & \begin{tikzpicture}[scale=0.8]
    \begin{groupplot}[
        group style={
            group size=2 by 2,
            group name=samples,
            x descriptions at=edge bottom,
            y descriptions at=edge left,
            horizontal sep=2mm,
            vertical sep=2mm
        },
        xlabel=Sample Size,
        ymin=-.1, ymax=1.1,
        grid=both, xmode=log, log basis x=2,
        xlabel style={yshift=1mm},
        ylabel style={yshift=-1mm},
        width=5cm
    ]
    
        \nextgroupplot[ylabel=Precision]
        \addplot[mark=o, smooth, thick, color=color0] table[x=samples, y=dagma_adj_pre]{results/er_20.txt}; 
        \addplot[mark=square, smooth, thick, color=color1] table[x=samples, y=boss_adj_pre]{results/er_20.txt}; 
        \addplot[mark=+, smooth, thick, color=color2] table[x=samples, y=fges_adj_pre]{results/er_20.txt}; 
        \addplot[mark=triangle, smooth, thick, color=color3] table[x=samples, y=pc_adj_pre]{results/er_20.txt}; 
        \addplot[mark=Mercedes star, smooth, thick, color=color4] table[x=samples, y=r2sr_adj_pre]{results/er_20.txt}; 
        \addplot[mark=diamond, smooth, thick, color=color5] table[x=samples, y=dlingam_adj_pre]{results/er_20.txt}; 
        \addplot[mark=x, smooth, thick, color=color6] table[x=samples, y=itmcmc_adj_pre]{results/er_20.txt}; 
        \addplot[mark=star, smooth, thick, color=color8] table[x=samples, y=mmhc_adj_pre]{results/er_20.txt}; 

        \nextgroupplot
        \addplot[mark=o, smooth, thick, color=color0] table[x=samples, y=dagma_ori_pre]{results/er_20.txt};
        \addplot[mark=square, smooth, thick, color=color1] table[x=samples, y=boss_ori_pre]{results/er_20.txt};
        \addplot[mark=+, smooth, thick, color=color2] table[x=samples, y=fges_ori_pre]{results/er_20.txt};
        \addplot[mark=triangle, smooth, thick, color=color3] table[x=samples, y=pc_ori_pre]{results/er_20.txt};
        \addplot[mark=Mercedes star, smooth, thick, color=color4] table[x=samples, y=r2sr_ori_pre]{results/er_20.txt};    
        \addplot[mark=diamond, smooth, thick, color=color5] table[x=samples, y=dlingam_ori_pre]{results/er_20.txt};
        \addplot[mark=x, smooth, thick, color=color6] table[x=samples, y=itmcmc_ori_pre]{results/er_20.txt};
        \addplot[mark=star, smooth, thick, color=color8] table[x=samples, y=mmhc_ori_pre]{results/er_20.txt};
        
        \nextgroupplot[ylabel=Recall]
        \addplot[mark=o, smooth, thick, color=color0] table[x=samples, y=dagma_adj_rec]{results/er_20.txt};
        \addplot[mark=square, smooth, thick, color=color1] table[x=samples, y=boss_adj_rec]{results/er_20.txt};
        \addplot[mark=+, smooth, thick, color=color2] table[x=samples, y=fges_adj_rec]{results/er_20.txt};
        \addplot[mark=triangle, smooth, thick, color=color3] table[x=samples, y=pc_adj_rec]{results/er_20.txt};
        \addplot[mark=Mercedes star, smooth, thick, color=color4] table[x=samples, y=r2sr_adj_rec]{results/er_20.txt};    
        \addplot[mark=diamond, smooth, thick, color=color5] table[x=samples, y=dlingam_adj_rec]{results/er_20.txt};
        \addplot[mark=x, smooth, thick, color=color6] table[x=samples, y=itmcmc_adj_rec]{results/er_20.txt};
        \addplot[mark=star, smooth, thick, color=color8] table[x=samples, y=mmhc_adj_rec]{results/er_20.txt};

        \nextgroupplot
        \addplot[mark=o, smooth, thick, color=color0] table[x=samples, y=dagma_ori_rec]{results/er_20.txt};
        \addplot[mark=square, smooth, thick, color=color1] table[x=samples, y=boss_ori_rec]{results/er_20.txt};
        \addplot[mark=+, smooth, thick, color=color2] table[x=samples, y=fges_ori_rec]{results/er_20.txt};
        \addplot[mark=triangle, smooth, thick, color=color3] table[x=samples, y=pc_ori_rec]{results/er_20.txt};
        \addplot[mark=Mercedes star, smooth, thick, color=color4] table[x=samples, y=r2sr_ori_rec]{results/er_20.txt};    
        \addplot[mark=diamond, smooth, thick, color=color5] table[x=samples, y=dlingam_ori_rec]{results/er_20.txt};
        \addplot[mark=x, smooth, thick, color=color6] table[x=samples, y=itmcmc_ori_rec]{results/er_20.txt};
        \addplot[mark=star, smooth, thick, color=color8] table[x=samples, y=mmhc_ori_rec]{results/er_20.txt};
        
    \end{groupplot}
    \node[above=1mm of samples c1r1, scale=0.8] {Adjacency};
    \node[above=1mm of samples c2r1, scale=0.8] {Orientation};
\end{tikzpicture} & \begin{tikzpicture}[scale=0.8]
    \begin{groupplot}[
        group style={
            group size=2 by 2,
            group name=samples,
            x descriptions at=edge bottom,
            y descriptions at=edge left,
            horizontal sep=2mm,
            vertical sep=2mm
        },
        xlabel=Sample Size,
        ymin=-.1, ymax=1.1,
        grid=both, xmode=log, log basis x=2,
        xlabel style={yshift=1mm},
        ylabel style={yshift=-1mm},
        width=5cm
    ]
    
        \nextgroupplot
        \addplot[mark=o, smooth, thick, color=color0] table[x=samples, y=dagma_adj_pre]{results/er_100.txt}; 
        \addplot[mark=square, smooth, thick, color=color1] table[x=samples, y=boss_adj_pre]{results/er_100.txt}; 
        \addplot[mark=+, smooth, thick, color=color2] table[x=samples, y=fges_adj_pre]{results/er_100.txt}; 
        \addplot[mark=triangle, smooth, thick, color=color3] table[x=samples, y=pc_adj_pre]{results/er_100.txt}; 
        \addplot[mark=Mercedes star, smooth, thick, color=color4] table[x=samples, y=r2sr_adj_pre]{results/er_100.txt}; 
        \addplot[mark=diamond, smooth, thick, color=color5] table[x=samples, y=dlingam_adj_pre]{results/er_100.txt}; 
        \addplot[mark=star, smooth, thick, color=color8] table[x=samples, y=mmhc_adj_pre]{results/er_100.txt}; 
        
        \nextgroupplot
        \addplot[mark=o, smooth, thick, color=color0] table[x=samples, y=dagma_ori_pre]{results/er_100.txt};
        \addplot[mark=square, smooth, thick, color=color1] table[x=samples, y=boss_ori_pre]{results/er_100.txt};
        \addplot[mark=+, smooth, thick, color=color2] table[x=samples, y=fges_ori_pre]{results/er_100.txt};
        \addplot[mark=triangle, smooth, thick, color=color3] table[x=samples, y=pc_ori_pre]{results/er_100.txt};
        \addplot[mark=Mercedes star, smooth, thick, color=color4] table[x=samples, y=r2sr_ori_pre]{results/er_100.txt};    
        \addplot[mark=diamond, smooth, thick, color=color5] table[x=samples, y=dlingam_ori_pre]{results/er_100.txt};
        \addplot[mark=star, smooth, thick, color=color8] table[x=samples, y=mmhc_ori_pre]{results/er_100.txt};
        
        \nextgroupplot
        \addplot[mark=o, smooth, thick, color=color0] table[x=samples, y=dagma_adj_rec]{results/er_100.txt};
        \addplot[mark=square, smooth, thick, color=color1] table[x=samples, y=boss_adj_rec]{results/er_100.txt};
        \addplot[mark=+, smooth, thick, color=color2] table[x=samples, y=fges_adj_rec]{results/er_100.txt};
        \addplot[mark=triangle, smooth, thick, color=color3] table[x=samples, y=pc_adj_rec]{results/er_100.txt};
        \addplot[mark=Mercedes star, smooth, thick, color=color4] table[x=samples, y=r2sr_adj_rec]{results/er_100.txt};    
        \addplot[mark=diamond, smooth, thick, color=color5] table[x=samples, y=dlingam_adj_rec]{results/er_100.txt};
        \addplot[mark=star, smooth, thick, color=color8] table[x=samples, y=mmhc_adj_rec]{results/er_100.txt};
        
        \nextgroupplot
        \addplot[mark=o, smooth, thick, color=color0] table[x=samples, y=dagma_ori_rec]{results/er_100.txt};
        \addplot[mark=square, smooth, thick, color=color1] table[x=samples, y=boss_ori_rec]{results/er_100.txt};
        \addplot[mark=+, smooth, thick, color=color2] table[x=samples, y=fges_ori_rec]{results/er_100.txt};
        \addplot[mark=triangle, smooth, thick, color=color3] table[x=samples, y=pc_ori_rec]{results/er_100.txt};
        \addplot[mark=Mercedes star, smooth, thick, color=color4] table[x=samples, y=r2sr_ori_rec]{results/er_100.txt};    
        \addplot[mark=diamond, smooth, thick, color=color5] table[x=samples, y=dlingam_ori_rec]{results/er_100.txt};
        \addplot[mark=star, smooth, thick, color=color8] table[x=samples, y=mmhc_ori_rec]{results/er_100.txt};
        
    \end{groupplot}
    \node[above=1mm of samples c1r1, scale=0.8] {Adjacency};
    \node[above=1mm of samples c2r1, scale=0.8] {Orientation};
\end{tikzpicture} \\
            \vphantom{\bigg|} & \multicolumn{2}{c}{\begin{tikzpicture}[scale=1.0]
    \node[fill=white, draw=white, scale=1.0] at (0, 0) {
        \small
        \begin{tabular}{ccccccc}
            \multicolumn{7}{c}{\vphantom{\Big|}Algorithm:} \\
            \ref*{er1} DAGMA & \s[1] & \ref*{er2} BOSS & \s[1] & \ref*{er3} fGES & \s[1] & \ref*{er4} PC \\
            \ref*{er5} $R^2$-sort & \s[1] & \ref*{er6} dLiNGAM & \s[1] & \ref*{er7} itMCMC & \s[1] & \ref*{er8} MMHC
        \end{tabular}
    };
\end{tikzpicture}}
        \end{tabular}
        \caption{Mean performance of CDAs on data generated from ER-DAGs with $\alpha = 10$ using the DaO method for DAG Models over 20 repetitions.}
        \label{fig:sim_study_er}
    \end{figure}
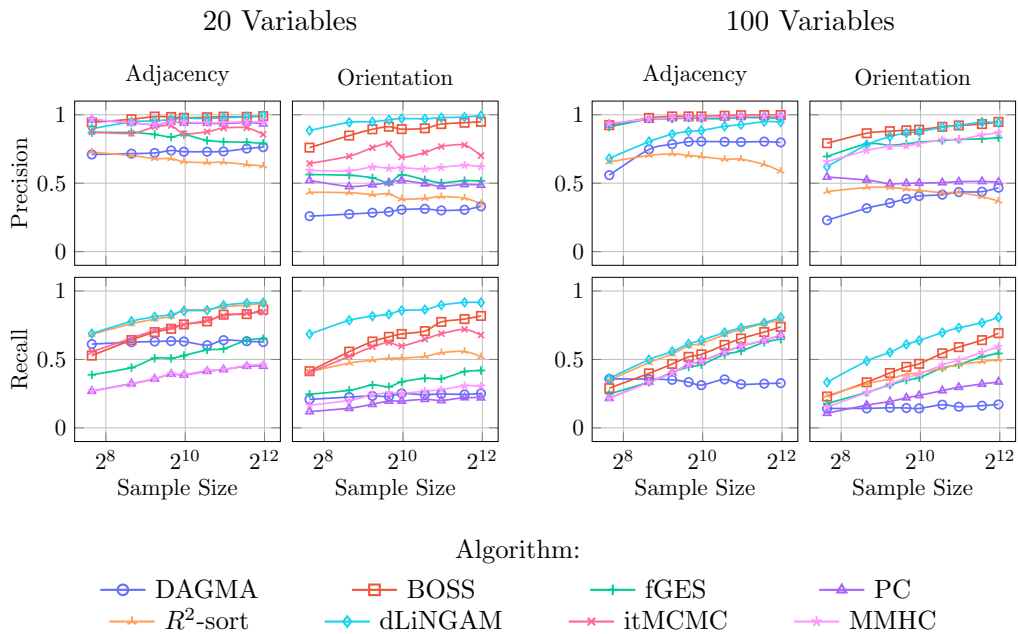

    \begin{figure}
        \centering
        \begin{tabular}{rcc}
            \vphantom{\bigg|} & \s[20] 20 Variables &  \s[20] 100 Variables \\
            & \begin{tikzpicture}[scale=0.8]
    \begin{groupplot}[
        group style={
            group size=2 by 2,
            group name=samples,
            x descriptions at=edge bottom,
            y descriptions at=edge left,
            horizontal sep=2mm,
            vertical sep=2mm
        },
        xlabel=Sample Size,
        ymin=-.1, ymax=1.1,
        grid=both, xmode=log, log basis x=2,
        xlabel style={yshift=1mm},
        ylabel style={yshift=-1mm},
        width=5cm
    ]
    
        \nextgroupplot[ylabel=Precision]
        \addplot[mark=o, smooth, thick, color=color0] table[x=samples, y=dagma_adj_pre]{results/sfi_20.txt}; 
        \addplot[mark=square, smooth, thick, color=color1] table[x=samples, y=boss_adj_pre]{results/sfi_20.txt}; 
        \addplot[mark=+, smooth, thick, color=color2] table[x=samples, y=fges_adj_pre]{results/sfi_20.txt}; 
        \addplot[mark=triangle, smooth, thick, color=color3] table[x=samples, y=pc_adj_pre]{results/sfi_20.txt}; 
        \addplot[mark=Mercedes star, smooth, thick, color=color4] table[x=samples, y=r2sr_adj_pre]{results/sfi_20.txt}; 
        \addplot[mark=diamond, smooth, thick, color=color5] table[x=samples, y=dlingam_adj_pre]{results/sfi_20.txt}; 
        \addplot[mark=x, smooth, thick, color=color6] table[x=samples, y=itmcmc_adj_pre]{results/sfi_20.txt}; 
        \addplot[mark=star, smooth, thick, color=color8] table[x=samples, y=mmhc_adj_pre]{results/sfi_20.txt}; 

        \nextgroupplot
        \addplot[mark=o, smooth, thick, color=color0] table[x=samples, y=dagma_ori_pre]{results/sfi_20.txt};
        \addplot[mark=square, smooth, thick, color=color1] table[x=samples, y=boss_ori_pre]{results/sfi_20.txt};
        \addplot[mark=+, smooth, thick, color=color2] table[x=samples, y=fges_ori_pre]{results/sfi_20.txt};
        \addplot[mark=triangle, smooth, thick, color=color3] table[x=samples, y=pc_ori_pre]{results/sfi_20.txt};
        \addplot[mark=Mercedes star, smooth, thick, color=color4] table[x=samples, y=r2sr_ori_pre]{results/sfi_20.txt};    
        \addplot[mark=diamond, smooth, thick, color=color5] table[x=samples, y=dlingam_ori_pre]{results/sfi_20.txt};
        \addplot[mark=x, smooth, thick, color=color6] table[x=samples, y=itmcmc_ori_pre]{results/sfi_20.txt};
        \addplot[mark=star, smooth, thick, color=color8] table[x=samples, y=mmhc_ori_pre]{results/sfi_20.txt};
        
        \nextgroupplot[ylabel=Recall]
        \addplot[mark=o, smooth, thick, color=color0] table[x=samples, y=dagma_adj_rec]{results/sfi_20.txt};
        \addplot[mark=square, smooth, thick, color=color1] table[x=samples, y=boss_adj_rec]{results/sfi_20.txt};
        \addplot[mark=+, smooth, thick, color=color2] table[x=samples, y=fges_adj_rec]{results/sfi_20.txt};
        \addplot[mark=triangle, smooth, thick, color=color3] table[x=samples, y=pc_adj_rec]{results/sfi_20.txt};
        \addplot[mark=Mercedes star, smooth, thick, color=color4] table[x=samples, y=r2sr_adj_rec]{results/sfi_20.txt};    
        \addplot[mark=diamond, smooth, thick, color=color5] table[x=samples, y=dlingam_adj_rec]{results/sfi_20.txt};
        \addplot[mark=x, smooth, thick, color=color6] table[x=samples, y=itmcmc_adj_rec]{results/sfi_20.txt};
        \addplot[mark=star, smooth, thick, color=color8] table[x=samples, y=mmhc_adj_rec]{results/sfi_20.txt};

        \nextgroupplot
        \addplot[mark=o, smooth, thick, color=color0] table[x=samples, y=dagma_ori_rec]{results/sfi_20.txt};
        \addplot[mark=square, smooth, thick, color=color1] table[x=samples, y=boss_ori_rec]{results/sfi_20.txt};
        \addplot[mark=+, smooth, thick, color=color2] table[x=samples, y=fges_ori_rec]{results/sfi_20.txt};
        \addplot[mark=triangle, smooth, thick, color=color3] table[x=samples, y=pc_ori_rec]{results/sfi_20.txt};
        \addplot[mark=Mercedes star, smooth, thick, color=color4] table[x=samples, y=r2sr_ori_rec]{results/sfi_20.txt};    
        \addplot[mark=diamond, smooth, thick, color=color5] table[x=samples, y=dlingam_ori_rec]{results/sfi_20.txt};
        \addplot[mark=x, smooth, thick, color=color6] table[x=samples, y=itmcmc_ori_rec]{results/sfi_20.txt};
        \addplot[mark=star, smooth, thick, color=color8] table[x=samples, y=mmhc_ori_rec]{results/sfi_20.txt};
        
    \end{groupplot}
    \node[above=1mm of samples c1r1, scale=0.8] {Adjacency};
    \node[above=1mm of samples c2r1, scale=0.8] {Orientation};
\end{tikzpicture} & \begin{tikzpicture}[scale=0.8]
    \begin{groupplot}[
        group style={
            group size=2 by 2,
            group name=samples,
            x descriptions at=edge bottom,
            y descriptions at=edge left,
            horizontal sep=2mm,
            vertical sep=2mm
        },
        xlabel=Sample Size,
        ymin=-.1, ymax=1.1,
        grid=both, xmode=log, log basis x=2,
        xlabel style={yshift=1mm},
        ylabel style={yshift=-1mm},
        width=5cm
    ]
    
        \nextgroupplot
        \addplot[mark=o, smooth, thick, color=color0] table[x=samples, y=dagma_adj_pre]{results/sfi_100.txt}; 
        \addplot[mark=square, smooth, thick, color=color1] table[x=samples, y=boss_adj_pre]{results/sfi_100.txt}; 
        \addplot[mark=+, smooth, thick, color=color2] table[x=samples, y=fges_adj_pre]{results/sfi_100.txt}; 
        \addplot[mark=Mercedes star, smooth, thick, color=color4] table[x=samples, y=r2sr_adj_pre]{results/sfi_100.txt}; 
        \addplot[mark=diamond, smooth, thick, color=color5] table[x=samples, y=dlingam_adj_pre]{results/sfi_100.txt}; 
        \addplot[mark=star, smooth, thick, color=color8] table[x=samples, y=mmhc_adj_pre]{results/sfi_100.txt}; 
        
        \nextgroupplot
        \addplot[mark=o, smooth, thick, color=color0] table[x=samples, y=dagma_ori_pre]{results/sfi_100.txt};
        \addplot[mark=square, smooth, thick, color=color1] table[x=samples, y=boss_ori_pre]{results/sfi_100.txt};
        \addplot[mark=+, smooth, thick, color=color2] table[x=samples, y=fges_ori_pre]{results/sfi_100.txt};
        \addplot[mark=Mercedes star, smooth, thick, color=color4] table[x=samples, y=r2sr_ori_pre]{results/sfi_100.txt};    
        \addplot[mark=diamond, smooth, thick, color=color5] table[x=samples, y=dlingam_ori_pre]{results/sfi_100.txt};
        \addplot[mark=star, smooth, thick, color=color8] table[x=samples, y=mmhc_ori_pre]{results/sfi_100.txt};
        
        \nextgroupplot
        \addplot[mark=o, smooth, thick, color=color0] table[x=samples, y=dagma_adj_rec]{results/sfi_100.txt};
        \addplot[mark=square, smooth, thick, color=color1] table[x=samples, y=boss_adj_rec]{results/sfi_100.txt};
        \addplot[mark=+, smooth, thick, color=color2] table[x=samples, y=fges_adj_rec]{results/sfi_100.txt};
        \addplot[mark=Mercedes star, smooth, thick, color=color4] table[x=samples, y=r2sr_adj_rec]{results/sfi_100.txt};    
        \addplot[mark=diamond, smooth, thick, color=color5] table[x=samples, y=dlingam_adj_rec]{results/sfi_100.txt};
        \addplot[mark=star, smooth, thick, color=color8] table[x=samples, y=mmhc_adj_rec]{results/sfi_100.txt};
        
        \nextgroupplot
        \addplot[mark=o, smooth, thick, color=color0] table[x=samples, y=dagma_ori_rec]{results/sfi_100.txt};
        \addplot[mark=square, smooth, thick, color=color1] table[x=samples, y=boss_ori_rec]{results/sfi_100.txt};
        \addplot[mark=+, smooth, thick, color=color2] table[x=samples, y=fges_ori_rec]{results/sfi_100.txt};
        \addplot[mark=Mercedes star, smooth, thick, color=color4] table[x=samples, y=r2sr_ori_rec]{results/sfi_100.txt};    
        \addplot[mark=diamond, smooth, thick, color=color5] table[x=samples, y=dlingam_ori_rec]{results/sfi_100.txt};
        \addplot[mark=star, smooth, thick, color=color8] table[x=samples, y=mmhc_ori_rec]{results/sfi_100.txt};
        
    \end{groupplot}
    \node[above=1mm of samples c1r1, scale=0.8] {Adjacency};
    \node[above=1mm of samples c2r1, scale=0.8] {Orientation};
\end{tikzpicture} \\
            \vphantom{\bigg|} & \multicolumn{2}{c}{\begin{tikzpicture}[scale=1.0]
    \node[fill=white, draw=white, scale=1.0] at (0, 0) {
        \small
        \begin{tabular}{ccccccc}
            \multicolumn{7}{c}{\vphantom{\Big|}Algorithm:} \\
            \ref*{er1} DAGMA & \s[1] & \ref*{er2} BOSS & \s[1] & \ref*{er3} fGES & \s[1] & \ref*{er4} PC \\
            \ref*{er5} $R^2$-sort & \s[1] & \ref*{er6} dLiNGAM & \s[1] & \ref*{er7} itMCMC & \s[1] & \ref*{er8} MMHC
        \end{tabular}
    };
\end{tikzpicture}}
        \end{tabular}
        \caption{Mean performance of CDAs on data generated from SFi-DAGs with $\alpha = 10$ using the DaO method for DAG Models over 20 repetitions.}
        \label{fig:sim_study_sfi}
    \end{figure}
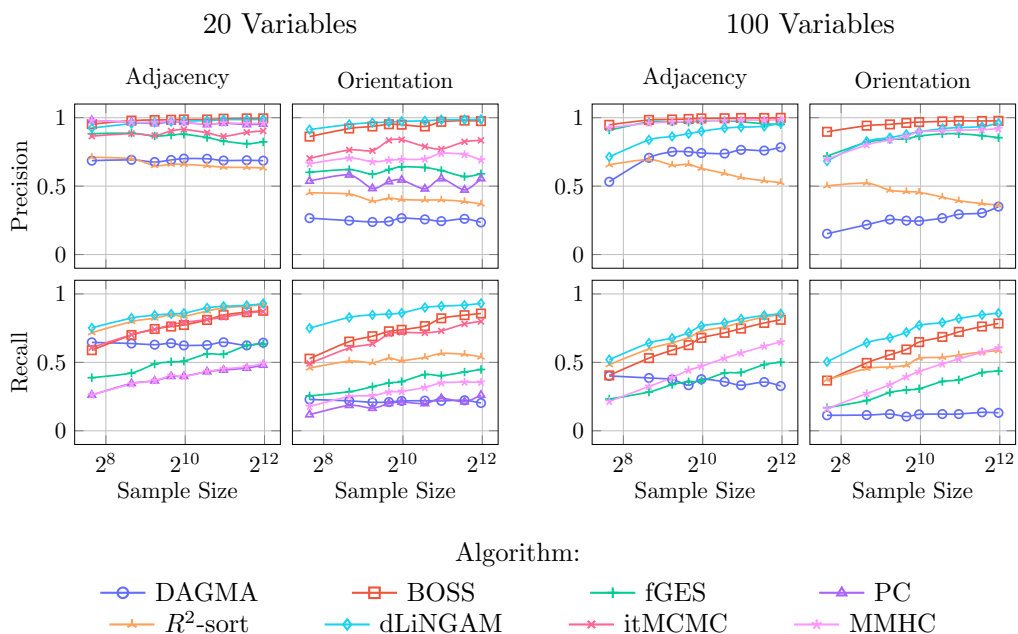

    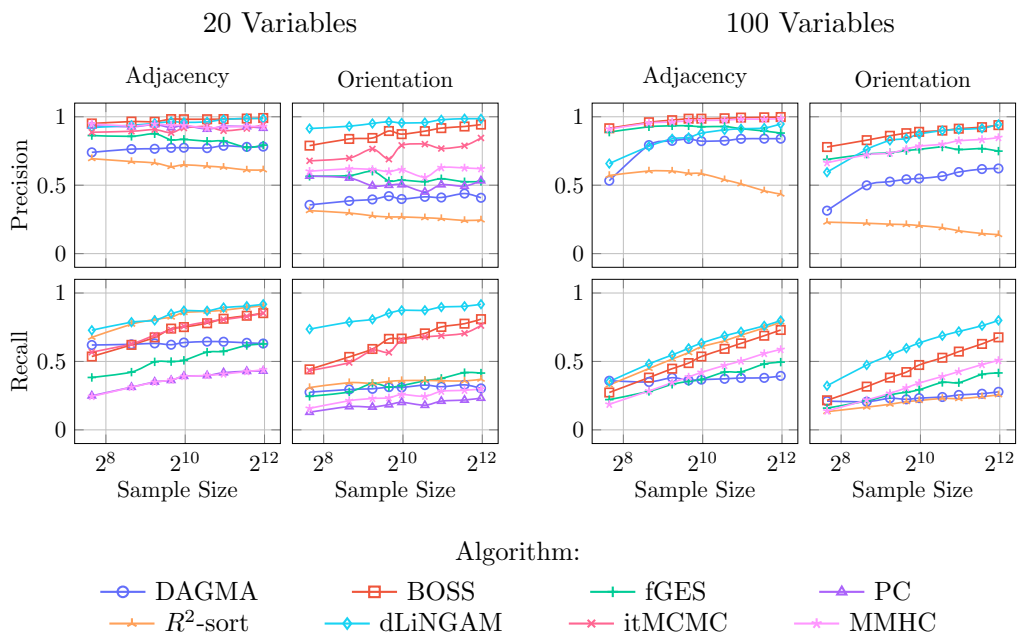
\begin{figure}
        \centering
        \begin{tabular}{rcc}
            \vphantom{\bigg|} & \s[20] 20 Variables &  \s[20] 100 Variables \\
            & \begin{tikzpicture}[scale=0.8]
    \begin{groupplot}[
        group style={
            group size=2 by 2,
            group name=samples,
            x descriptions at=edge bottom,
            y descriptions at=edge left,
            horizontal sep=2mm,
            vertical sep=2mm
        },
        xlabel=Sample Size,
        ymin=-.1, ymax=1.1,
        grid=both, xmode=log, log basis x=2,
        xlabel style={yshift=1mm},
        ylabel style={yshift=-1mm},
        width=5cm
    ]
    
        \nextgroupplot[ylabel=Precision]
        \addplot[mark=o, smooth, thick, color=color0] table[x=samples, y=dagma_adj_pre]{results/sfo_20.txt}; 
        \addplot[mark=square, smooth, thick, color=color1] table[x=samples, y=boss_adj_pre]{results/sfo_20.txt}; 
        \addplot[mark=+, smooth, thick, color=color2] table[x=samples, y=fges_adj_pre]{results/sfo_20.txt}; 
        \addplot[mark=triangle, smooth, thick, color=color3] table[x=samples, y=pc_adj_pre]{results/sfo_20.txt}; 
        \addplot[mark=Mercedes star, smooth, thick, color=color4] table[x=samples, y=r2sr_adj_pre]{results/sfo_20.txt}; 
        \addplot[mark=diamond, smooth, thick, color=color5] table[x=samples, y=dlingam_adj_pre]{results/sfo_20.txt}; 
        \addplot[mark=x, smooth, thick, color=color6] table[x=samples, y=itmcmc_adj_pre]{results/sfo_20.txt}; 
        \addplot[mark=star, smooth, thick, color=color8] table[x=samples, y=mmhc_adj_pre]{results/sfo_20.txt}; 

        \nextgroupplot
        \addplot[mark=o, smooth, thick, color=color0] table[x=samples, y=dagma_ori_pre]{results/sfo_20.txt};
        \addplot[mark=square, smooth, thick, color=color1] table[x=samples, y=boss_ori_pre]{results/sfo_20.txt};
        \addplot[mark=+, smooth, thick, color=color2] table[x=samples, y=fges_ori_pre]{results/sfo_20.txt};
        \addplot[mark=triangle, smooth, thick, color=color3] table[x=samples, y=pc_ori_pre]{results/sfo_20.txt};
        \addplot[mark=Mercedes star, smooth, thick, color=color4] table[x=samples, y=r2sr_ori_pre]{results/sfo_20.txt};    
        \addplot[mark=diamond, smooth, thick, color=color5] table[x=samples, y=dlingam_ori_pre]{results/sfo_20.txt};
        \addplot[mark=x, smooth, thick, color=color6] table[x=samples, y=itmcmc_ori_pre]{results/sfo_20.txt};
        \addplot[mark=star, smooth, thick, color=color8] table[x=samples, y=mmhc_ori_pre]{results/sfo_20.txt};
        
        \nextgroupplot[ylabel=Recall]
        \addplot[mark=o, smooth, thick, color=color0] table[x=samples, y=dagma_adj_rec]{results/sfo_20.txt};
        \addplot[mark=square, smooth, thick, color=color1] table[x=samples, y=boss_adj_rec]{results/sfo_20.txt};
        \addplot[mark=+, smooth, thick, color=color2] table[x=samples, y=fges_adj_rec]{results/sfo_20.txt};
        \addplot[mark=triangle, smooth, thick, color=color3] table[x=samples, y=pc_adj_rec]{results/sfo_20.txt};
        \addplot[mark=Mercedes star, smooth, thick, color=color4] table[x=samples, y=r2sr_adj_rec]{results/sfo_20.txt};    
        \addplot[mark=diamond, smooth, thick, color=color5] table[x=samples, y=dlingam_adj_rec]{results/sfo_20.txt};
        \addplot[mark=x, smooth, thick, color=color6] table[x=samples, y=itmcmc_adj_rec]{results/sfo_20.txt};
        \addplot[mark=star, smooth, thick, color=color8] table[x=samples, y=mmhc_adj_rec]{results/sfo_20.txt};

        \nextgroupplot
        \addplot[mark=o, smooth, thick, color=color0] table[x=samples, y=dagma_ori_rec]{results/sfo_20.txt};
        \addplot[mark=square, smooth, thick, color=color1] table[x=samples, y=boss_ori_rec]{results/sfo_20.txt};
        \addplot[mark=+, smooth, thick, color=color2] table[x=samples, y=fges_ori_rec]{results/sfo_20.txt};
        \addplot[mark=triangle, smooth, thick, color=color3] table[x=samples, y=pc_ori_rec]{results/sfo_20.txt};
        \addplot[mark=Mercedes star, smooth, thick, color=color4] table[x=samples, y=r2sr_ori_rec]{results/sfo_20.txt};    
        \addplot[mark=diamond, smooth, thick, color=color5] table[x=samples, y=dlingam_ori_rec]{results/sfo_20.txt};
        \addplot[mark=x, smooth, thick, color=color6] table[x=samples, y=itmcmc_ori_rec]{results/sfo_20.txt};
        \addplot[mark=star, smooth, thick, color=color8] table[x=samples, y=mmhc_ori_rec]{results/sfo_20.txt};
        
    \end{groupplot}
    \node[above=1mm of samples c1r1, scale=0.8] {Adjacency};
    \node[above=1mm of samples c2r1, scale=0.8] {Orientation};
\end{tikzpicture} & \begin{tikzpicture}[scale=0.8]
    \begin{groupplot}[
        group style={
            group size=2 by 2,
            group name=samples,
            x descriptions at=edge bottom,
            y descriptions at=edge left,
            horizontal sep=2mm,
            vertical sep=2mm
        },
        xlabel=Sample Size,
        ymin=-.1, ymax=1.1,
        grid=both, xmode=log, log basis x=2,
        xlabel style={yshift=1mm},
        ylabel style={yshift=-1mm},
        width=5cm
    ]
    
        \nextgroupplot
        \addplot[mark=o, smooth, thick, color=color0] table[x=samples, y=dagma_adj_pre]{results/sfo_100.txt}; 
        \addplot[mark=square, smooth, thick, color=color1] table[x=samples, y=boss_adj_pre]{results/sfo_100.txt}; 
        \addplot[mark=+, smooth, thick, color=color2] table[x=samples, y=fges_adj_pre]{results/sfo_100.txt}; 
        \addplot[mark=Mercedes star, smooth, thick, color=color4] table[x=samples, y=r2sr_adj_pre]{results/sfo_100.txt}; 
        \addplot[mark=diamond, smooth, thick, color=color5] table[x=samples, y=dlingam_adj_pre]{results/sfo_100.txt}; 
        \addplot[mark=star, smooth, thick, color=color8] table[x=samples, y=mmhc_adj_pre]{results/sfo_100.txt}; 
        
        \nextgroupplot
        \addplot[mark=o, smooth, thick, color=color0] table[x=samples, y=dagma_ori_pre]{results/sfo_100.txt};
        \addplot[mark=square, smooth, thick, color=color1] table[x=samples, y=boss_ori_pre]{results/sfo_100.txt};
        \addplot[mark=+, smooth, thick, color=color2] table[x=samples, y=fges_ori_pre]{results/sfo_100.txt};
        \addplot[mark=Mercedes star, smooth, thick, color=color4] table[x=samples, y=r2sr_ori_pre]{results/sfo_100.txt};    
        \addplot[mark=diamond, smooth, thick, color=color5] table[x=samples, y=dlingam_ori_pre]{results/sfo_100.txt};
        \addplot[mark=star, smooth, thick, color=color8] table[x=samples, y=mmhc_ori_pre]{results/sfo_100.txt};
        
        \nextgroupplot
        \addplot[mark=o, smooth, thick, color=color0] table[x=samples, y=dagma_adj_rec]{results/sfo_100.txt};
        \addplot[mark=square, smooth, thick, color=color1] table[x=samples, y=boss_adj_rec]{results/sfo_100.txt};
        \addplot[mark=+, smooth, thick, color=color2] table[x=samples, y=fges_adj_rec]{results/sfo_100.txt};
        \addplot[mark=Mercedes star, smooth, thick, color=color4] table[x=samples, y=r2sr_adj_rec]{results/sfo_100.txt};    
        \addplot[mark=diamond, smooth, thick, color=color5] table[x=samples, y=dlingam_adj_rec]{results/sfo_100.txt};
        \addplot[mark=star, smooth, thick, color=color8] table[x=samples, y=mmhc_adj_rec]{results/sfo_100.txt};
        
        \nextgroupplot
        \addplot[mark=o, smooth, thick, color=color0] table[x=samples, y=dagma_ori_rec]{results/sfo_100.txt};
        \addplot[mark=square, smooth, thick, color=color1] table[x=samples, y=boss_ori_rec]{results/sfo_100.txt};
        \addplot[mark=+, smooth, thick, color=color2] table[x=samples, y=fges_ori_rec]{results/sfo_100.txt};
        \addplot[mark=Mercedes star, smooth, thick, color=color4] table[x=samples, y=r2sr_ori_rec]{results/sfo_100.txt};    
        \addplot[mark=diamond, smooth, thick, color=color5] table[x=samples, y=dlingam_ori_rec]{results/sfo_100.txt};
        \addplot[mark=star, smooth, thick, color=color8] table[x=samples, y=mmhc_ori_rec]{results/sfo_100.txt};
        
    \end{groupplot}
    \node[above=1mm of samples c1r1, scale=0.8] {Adjacency};
    \node[above=1mm of samples c2r1, scale=0.8] {Orientation};
\end{tikzpicture} \\
            \vphantom{\bigg|} & \multicolumn{2}{c}{\begin{tikzpicture}[scale=1.0]
    \node[fill=white, draw=white, scale=1.0] at (0, 0) {
        \small
        \begin{tabular}{ccccccc}
            \multicolumn{7}{c}{\vphantom{\Big|}Algorithm:} \\
            \ref*{er1} DAGMA & \s[1] & \ref*{er2} BOSS & \s[1] & \ref*{er3} fGES & \s[1] & \ref*{er4} PC \\
            \ref*{er5} $R^2$-sort & \s[1] & \ref*{er6} dLiNGAM & \s[1] & \ref*{er7} itMCMC & \s[1] & \ref*{er8} MMHC
        \end{tabular}
    };
\end{tikzpicture}}
        \end{tabular}
        \caption{Mean performance of CDAs on data generated from SFo-DAGs with $\alpha = 10$ using the DaO method for DAG Models over 20 repetitions.}
        \label{fig:sim_study_sfo}
    \end{figure}

\subsection{Alternative Simulations}
\label{sec:alt_sims}

    In the alternative simulations, all graphs were generated using the ER-DAG method. The model parameters were used to generate two versions of the simulated data: Gaussian error and non-Gaussian exponential error. The Gaussian version of the simulated data was used to evaluate all CDAs with one exception. The non-Gaussian version of the simulated data was used to evaluate dLiNGAM due to its dependence on a non-Gaussian signal. Accordingly, dLiNGAM had access to additional information that the other algorithms did not, which is reflected in its superior performance. The simulated data were not standardized before running any of the CDAs with one exception. DAGMA was run twice: once on the unstandardized data (DAGMA) and once on the standardized data (sDAGMA). 

    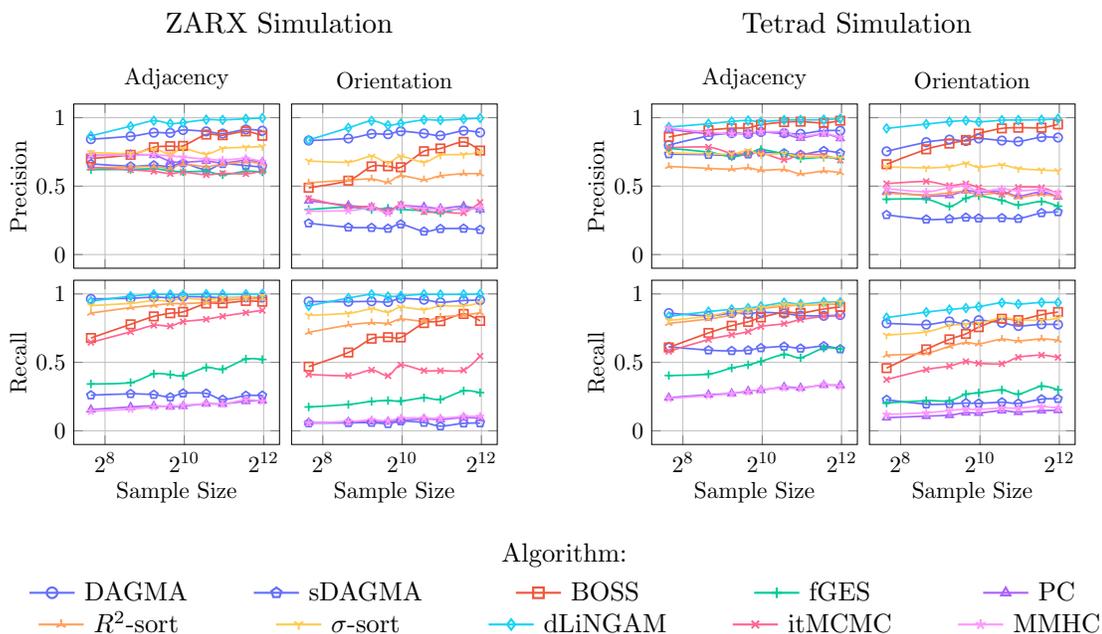
\begin{figure}
        \centering
        \begin{tabular}{rcc}
            \vphantom{\bigg|} & \s[20] ZARX Simulation &  \s[20] Tetrad Simulation \\
            & \begin{tikzpicture}[scale=0.8]
    \begin{groupplot}[
        group style={
            group size=2 by 2,
            group name=samples,
            x descriptions at=edge bottom,
            y descriptions at=edge left,
            horizontal sep=2mm,
            vertical sep=2mm
        },
        xlabel=Sample Size,
        ymin=-.1, ymax=1.1,
        grid=both, xmode=log, log basis x=2,
        xlabel style={yshift=1mm},
        ylabel style={yshift=-1mm},
        width=5cm
    ]
    
        \nextgroupplot[ylabel=Precision]
        \addplot[mark=o, smooth, thick, color=color0] table[x=samples, y=dagma_adj_pre]{results/dagma_20.txt}; \label{er1}
        \addplot[mark=pentagon, smooth, thick, color=color0] table[x=samples, y=sdagma_adj_pre]{results/dagma_20.txt}; \label{er10}
        \addplot[mark=square, smooth, thick, color=color1] table[x=samples, y=boss_adj_pre]{results/dagma_20.txt}; \label{er2}
        \addplot[mark=+, smooth, thick, color=color2] table[x=samples, y=fges_adj_pre]{results/dagma_20.txt}; \label{er3}
        \addplot[mark=triangle, smooth, thick, color=color3] table[x=samples, y=pc_adj_pre]{results/dagma_20.txt}; \label{er4}
        \addplot[mark=Mercedes star, smooth, thick, color=color4] table[x=samples, y=r2sr_adj_pre]{results/dagma_20.txt}; \label{er5}
        \addplot[mark=diamond, smooth, thick, color=color5] table[x=samples, y=dlingam_adj_pre]{results/dagma_20.txt}; \label{er6}
        \addplot[mark=x, smooth, thick, color=color6] table[x=samples, y=itmcmc_adj_pre]{results/dagma_20.txt}; \label{er7}
        \addplot[mark=star, smooth, thick, color=color8] table[x=samples, y=mmhc_adj_pre]{results/dagma_20.txt}; \label{er8}
        \addplot[mark=Mercedes star flipped, smooth, thick, color=color9] table[x=samples, y=varsr_adj_pre]{results/dagma_20.txt}; \label{er9}

        \nextgroupplot
        \addplot[mark=o, smooth, thick, color=color0] table[x=samples, y=dagma_ori_pre]{results/dagma_20.txt};
        \addplot[mark=pentagon, smooth, thick, color=color0] table[x=samples, y=sdagma_ori_pre]{results/dagma_20.txt};
        \addplot[mark=square, smooth, thick, color=color1] table[x=samples, y=boss_ori_pre]{results/dagma_20.txt};
        \addplot[mark=+, smooth, thick, color=color2] table[x=samples, y=fges_ori_pre]{results/dagma_20.txt};
        \addplot[mark=triangle, smooth, thick, color=color3] table[x=samples, y=pc_ori_pre]{results/dagma_20.txt};
        \addplot[mark=Mercedes star, smooth, thick, color=color4] table[x=samples, y=r2sr_ori_pre]{results/dagma_20.txt};    
        \addplot[mark=diamond, smooth, thick, color=color5] table[x=samples, y=dlingam_ori_pre]{results/dagma_20.txt};
        \addplot[mark=x, smooth, thick, color=color6] table[x=samples, y=itmcmc_ori_pre]{results/dagma_20.txt};
        \addplot[mark=star, smooth, thick, color=color8] table[x=samples, y=mmhc_ori_pre]{results/dagma_20.txt};
        \addplot[mark=Mercedes star flipped, smooth, thick, color=color9] table[x=samples, y=varsr_ori_pre]{results/dagma_20.txt};    
        
        \nextgroupplot[ylabel=Recall]
        \addplot[mark=o, smooth, thick, color=color0] table[x=samples, y=dagma_adj_rec]{results/dagma_20.txt};
        \addplot[mark=pentagon, smooth, thick, color=color0] table[x=samples, y=sdagma_adj_rec]{results/dagma_20.txt};
        \addplot[mark=square, smooth, thick, color=color1] table[x=samples, y=boss_adj_rec]{results/dagma_20.txt};
        \addplot[mark=+, smooth, thick, color=color2] table[x=samples, y=fges_adj_rec]{results/dagma_20.txt};
        \addplot[mark=triangle, smooth, thick, color=color3] table[x=samples, y=pc_adj_rec]{results/dagma_20.txt};
        \addplot[mark=Mercedes star, smooth, thick, color=color4] table[x=samples, y=r2sr_adj_rec]{results/dagma_20.txt};    
        \addplot[mark=diamond, smooth, thick, color=color5] table[x=samples, y=dlingam_adj_rec]{results/dagma_20.txt};
        \addplot[mark=x, smooth, thick, color=color6] table[x=samples, y=itmcmc_adj_rec]{results/dagma_20.txt};
        \addplot[mark=star, smooth, thick, color=color8] table[x=samples, y=mmhc_adj_rec]{results/dagma_20.txt};
        \addplot[mark=Mercedes star flipped, smooth, thick, color=color9] table[x=samples, y=varsr_adj_rec]{results/dagma_20.txt};    

        \nextgroupplot
        \addplot[mark=o, smooth, thick, color=color0] table[x=samples, y=dagma_ori_rec]{results/dagma_20.txt};
        \addplot[mark=pentagon, smooth, thick, color=color0] table[x=samples, y=sdagma_ori_rec]{results/dagma_20.txt};
        \addplot[mark=square, smooth, thick, color=color1] table[x=samples, y=boss_ori_rec]{results/dagma_20.txt};
        \addplot[mark=+, smooth, thick, color=color2] table[x=samples, y=fges_ori_rec]{results/dagma_20.txt};
        \addplot[mark=triangle, smooth, thick, color=color3] table[x=samples, y=pc_ori_rec]{results/dagma_20.txt};
        \addplot[mark=Mercedes star, smooth, thick, color=color4] table[x=samples, y=r2sr_ori_rec]{results/dagma_20.txt};    
        \addplot[mark=diamond, smooth, thick, color=color5] table[x=samples, y=dlingam_ori_rec]{results/dagma_20.txt};
        \addplot[mark=x, smooth, thick, color=color6] table[x=samples, y=itmcmc_ori_rec]{results/dagma_20.txt};
        \addplot[mark=star, smooth, thick, color=color8] table[x=samples, y=mmhc_ori_rec]{results/dagma_20.txt};
        \addplot[mark=Mercedes star flipped, smooth, thick, color=color9] table[x=samples, y=varsr_ori_rec]{results/dagma_20.txt};    
        
    \end{groupplot}
    \node[above=1mm of samples c1r1, scale=0.8] {Adjacency};
    \node[above=1mm of samples c2r1, scale=0.8] {Orientation};
\end{tikzpicture} & \begin{tikzpicture}[scale=0.8]
    \begin{groupplot}[
        group style={
            group size=2 by 2,
            group name=samples,
            x descriptions at=edge bottom,
            y descriptions at=edge left,
            horizontal sep=2mm,
            vertical sep=2mm
        },
        xlabel=Sample Size,
        ymin=-.1, ymax=1.1,
        grid=both, xmode=log, log basis x=2,
        xlabel style={yshift=1mm},
        ylabel style={yshift=-1mm},
        width=5cm
    ]
    
        \nextgroupplot[ylabel=Precision]
        \addplot[mark=o, smooth, thick, color=color0] table[x=samples, y=dagma_adj_pre]{results/boss_20.txt}; 
        \addplot[mark=pentagon, smooth, thick, color=color0] table[x=samples, y=sdagma_adj_pre]{results/boss_20.txt}; 
        \addplot[mark=square, smooth, thick, color=color1] table[x=samples, y=boss_adj_pre]{results/boss_20.txt}; 
        \addplot[mark=+, smooth, thick, color=color2] table[x=samples, y=fges_adj_pre]{results/boss_20.txt}; 
        \addplot[mark=triangle, smooth, thick, color=color3] table[x=samples, y=pc_adj_pre]{results/boss_20.txt}; 
        \addplot[mark=Mercedes star, smooth, thick, color=color4] table[x=samples, y=r2sr_adj_pre]{results/boss_20.txt}; 
        \addplot[mark=diamond, smooth, thick, color=color5] table[x=samples, y=dlingam_adj_pre]{results/boss_20.txt}; 
        \addplot[mark=x, smooth, thick, color=color6] table[x=samples, y=itmcmc_adj_pre]{results/boss_20.txt}; 
        \addplot[mark=star, smooth, thick, color=color8] table[x=samples, y=mmhc_adj_pre]{results/boss_20.txt}; 
        \addplot[mark=Mercedes star flipped, smooth, thick, color=color9] table[x=samples, y=varsr_adj_pre]{results/boss_20.txt}; 

        \nextgroupplot
        \addplot[mark=o, smooth, thick, color=color0] table[x=samples, y=dagma_ori_pre]{results/boss_20.txt};
        \addplot[mark=pentagon, smooth, thick, color=color0] table[x=samples, y=sdagma_ori_pre]{results/boss_20.txt};
        \addplot[mark=square, smooth, thick, color=color1] table[x=samples, y=boss_ori_pre]{results/boss_20.txt};
        \addplot[mark=+, smooth, thick, color=color2] table[x=samples, y=fges_ori_pre]{results/boss_20.txt};
        \addplot[mark=triangle, smooth, thick, color=color3] table[x=samples, y=pc_ori_pre]{results/boss_20.txt};
        \addplot[mark=Mercedes star, smooth, thick, color=color4] table[x=samples, y=r2sr_ori_pre]{results/boss_20.txt};    
        \addplot[mark=diamond, smooth, thick, color=color5] table[x=samples, y=dlingam_ori_pre]{results/boss_20.txt};
        \addplot[mark=x, smooth, thick, color=color6] table[x=samples, y=itmcmc_ori_pre]{results/boss_20.txt};
        \addplot[mark=star, smooth, thick, color=color8] table[x=samples, y=mmhc_ori_pre]{results/boss_20.txt};
        \addplot[mark=Mercedes star flipped, smooth, thick, color=color9] table[x=samples, y=varsr_ori_pre]{results/boss_20.txt};    
        
        \nextgroupplot[ylabel=Recall]
        \addplot[mark=o, smooth, thick, color=color0] table[x=samples, y=dagma_adj_rec]{results/boss_20.txt};
        \addplot[mark=pentagon, smooth, thick, color=color0] table[x=samples, y=sdagma_adj_rec]{results/boss_20.txt};
        \addplot[mark=square, smooth, thick, color=color1] table[x=samples, y=boss_adj_rec]{results/boss_20.txt};
        \addplot[mark=+, smooth, thick, color=color2] table[x=samples, y=fges_adj_rec]{results/boss_20.txt};
        \addplot[mark=triangle, smooth, thick, color=color3] table[x=samples, y=pc_adj_rec]{results/boss_20.txt};
        \addplot[mark=Mercedes star, smooth, thick, color=color4] table[x=samples, y=r2sr_adj_rec]{results/boss_20.txt};    
        \addplot[mark=diamond, smooth, thick, color=color5] table[x=samples, y=dlingam_adj_rec]{results/boss_20.txt};
        \addplot[mark=x, smooth, thick, color=color6] table[x=samples, y=itmcmc_adj_rec]{results/boss_20.txt};
        \addplot[mark=star, smooth, thick, color=color8] table[x=samples, y=mmhc_adj_rec]{results/boss_20.txt};
        \addplot[mark=Mercedes star flipped, smooth, thick, color=color9] table[x=samples, y=varsr_adj_rec]{results/boss_20.txt};    

        \nextgroupplot
        \addplot[mark=o, smooth, thick, color=color0] table[x=samples, y=dagma_ori_rec]{results/boss_20.txt};
        \addplot[mark=pentagon, smooth, thick, color=color0] table[x=samples, y=sdagma_ori_rec]{results/boss_20.txt};
        \addplot[mark=square, smooth, thick, color=color1] table[x=samples, y=boss_ori_rec]{results/boss_20.txt};
        \addplot[mark=+, smooth, thick, color=color2] table[x=samples, y=fges_ori_rec]{results/boss_20.txt};
        \addplot[mark=triangle, smooth, thick, color=color3] table[x=samples, y=pc_ori_rec]{results/boss_20.txt};
        \addplot[mark=Mercedes star, smooth, thick, color=color4] table[x=samples, y=r2sr_ori_rec]{results/boss_20.txt};    
        \addplot[mark=diamond, smooth, thick, color=color5] table[x=samples, y=dlingam_ori_rec]{results/boss_20.txt};
        \addplot[mark=x, smooth, thick, color=color6] table[x=samples, y=itmcmc_ori_rec]{results/boss_20.txt};
        \addplot[mark=star, smooth, thick, color=color8] table[x=samples, y=mmhc_ori_rec]{results/boss_20.txt};
        \addplot[mark=Mercedes star flipped, smooth, thick, color=color9] table[x=samples, y=varsr_ori_rec]{results/boss_20.txt};    
        
    \end{groupplot}
    \node[above=1mm of samples c1r1, scale=0.8] {Adjacency};
    \node[above=1mm of samples c2r1, scale=0.8] {Orientation};
\end{tikzpicture} \\
            \vphantom{\bigg|} & \multicolumn{2}{c}{\begin{tikzpicture}[scale=1.0]
    \node[fill=white, draw=white, scale=1.0] at (0, 0) {
        \small
        \begin{tabular}{ccccccccc}
            \multicolumn{9}{c}{\vphantom{\Big|}Algorithm:} \\
            \ref*{er1} DAGMA & \s[1] & \ref*{er10} sDAGMA & \s[1] & \ref*{er2} BOSS & \s[1] & \ref*{er3} fGES & \s[1] & \ref*{er4} PC \\
            \ref*{er5} $R^2$-sort & \s[1] & \ref*{er9} $\sigma$-sort & \s[1] & \ref*{er6} dLiNGAM & \s[1] & \ref*{er7} itMCMC & \s[1] & \ref*{er8} MMHC
        \end{tabular}
    };
\end{tikzpicture}}
        \end{tabular}
        \caption{Mean performance of CDAs on data generated from ER-DAGs with 20 variables and $\alpha = 10$ over 20 repetitions.}
        \label{fig:sim_study_other}
    \end{figure}

    As noted in Section \ref{sec:eval_of_params}, ZARX simulations create near deterministic relationships between the variables. We see that methods such as PC and GES perform poorly in the presence of this near determinism. We conjecture this is because the near deterministic relationships result in (near-)violations of faithfulness and local optima. Moreover, the difference in performance between DAGMA and sDAGMA highlights the complains of several papers targeting the simulations commonly used to evaluate continuous optimization-based methods.

\section{Discussion}
\label{sec:discussion}

    We present new methods for simulating data for validating causal discovery; Python and R implementations are available: \url{https://github.com/bja43/DaO_simulation}. Given a DAG $\mc G$, these methods include: (i) DAG-sampling methods to modify $\mc G$ to have scale-free in-degree or out-degree, and (ii) the DaO method which samples uniformly from the space of correlation matrices satisfying the Markov property for $\mc G$. 
    
    The former methods are useful since many real-world systems are believed to have scale-free structure. The latter method is distinct compared to simulation methods used in past studies due to its generic nature and inherent fairness. It is generic because it draws models that uniformly cover the space of all possible correlation matrices. Further, it is fair because it does not rely on tuning parameters that can be manipulated to favor some algorithms over others. As a result, methods that perform well on DaO simulations will likely excel across many situations. 


    Several papers have criticized a common data simulation practice. These papers identified properties they claimed are unrealistic and could be unfairly exploited by learning algorithms. We have shown here that these properties are largely not present in parametric models generated by the DaO method. 
    
    Interestingly, $R^2$-sortability is still present, albeit weakly, in data generated by the DaO method. Since the DaO method samples models that uniformly cover the space of correlation matrices Markov to the DAG, this indicates that there may be something about DAG structures themselves that can produce $R^2$-sortability. This observation is supported by the dependence between the DAG structure sampling method (ER, SFo, SFi) and the rank correlation coefficient between each variable's $R^2$-sortability rank and its index from the variable order used to generate the DAG; see Tables \ref{tab:r2s_20} and \ref{tab:r2s_100}.

\subsection{Insights into the Conflicting Results of Previous Simulations}

    This paper also sheds some light on the ongoing debate surrounding continuous optimization-based methods.
    
    First, it highlights that continuous optimization-based methods need to pick a lower threshold in order to recover edges from standardized models. This is illustrated by Figures \ref{fig:er_sims}, \ref{fig:sfi_sims}, and \ref{fig:sfo_sims}. In addition to lacking signals like varsortability, this may be part of the reason that continuous optimization-based methods did not perform well on the DaO simulations reported in Figures \ref{fig:sim_study_er}, \ref{fig:sim_study_sfi}, and \ref{fig:sim_study_sfo}.

    Second, the need to bound beta coefficients away from zero is a misconception. Instead, this can result in near determinism which creates (near-)violations of faithfulness. Figures \ref{fig:er_sims}, \ref{fig:sfi_sims}, and \ref{fig:sfo_sims} highlight that the ZARX simulations, used by continuous optimization publications, produce a substantial number of these nearly deterministic relationships. Moreover, non-continuous optimization-based CDA are known to perform poorly in the presence of deterministic relationships \citep{glymour2007learning}

    Together, the above two points help to explain the substantial conflict between the simulation results reported in papers that use the ZARX method versus simulation results that use other methods; see the performance of DAGMA and sDAGMA in Figure \ref{fig:sim_study_other} for an illustration of this issue.



\subsection{Limitations}

    The DaO method has several limitations. First, the DaO method and DAG sampling methods in general are naturally limited to acyclic models, however, many causal processes in the natural world are cyclic. Second, we do not present any methodology for simulation or evaluation in the presence of unmeasured confounding. Third, our simulations only produce cross-sectional data and cannot be directly used for evaluating methods intended for time-series contexts. Fourth, while sampling uniformly from the space of all correlation matrices has many advantages, we may want to know how an algorithm is likely to perform on a specific dataset from a specific causal system. In such cases, a simulation that is more tailored to that specific dataset and causal system may be more informative than a DaO simulation.
    

\subsection{Future directions}

    
    Future directions include the following:
    \begin{enumerate}
        \item Extend the DaO method to latent variables.
        \item Extend the DaO method to time-series data.
        \item Design and execute a large simulation study using the DaO method.
        \item Investigate methods for tailoring the DaO method to specific ranges of correlation matrices, such as those similar to a correlation matrix provided as input.
        \item Contribute the DaO method to other tools that have been made for the comparison and validation of CDAs such as Benchpress \citep{rios2021benchpress}.
    \end{enumerate}

\acks{We thank Peter Spirtes and Joe Ramsey for insightful comments and discussion. BA was supported by the NIH under the Comorbidity: Substance Use Disorders and Other Psychiatric Conditions Training Program T32DA037183. EK was supported by the NIH under awards P50MH119569 and UL1TR002494.}

\bibliography{refs}

\begin{thebibliography}{59}
\providecommand{\natexlab}[1]{#1}
\providecommand{\url}[1]{\texttt{#1}}
\expandafter\ifx\csname urlstyle\endcsname\relax
  \providecommand{\doi}[1]{doi: #1}\else
  \providecommand{\doi}{doi: \begingroup \urlstyle{rm}\Url}\fi

\bibitem[Anderson(1984)]{anderson1984introduction}
T.W. Anderson.
\newblock \emph{An Introduction to Multivariate Statistical Analysis}.
\newblock An Introduction to Multivariate Statistical Analysis. Wiley, 1984.

\bibitem[Andrews et~al.(2023)Andrews, Ramsey, S{\'a}nchez-Romero, Camchong, and
  Kummerfeld]{andrews2023fast}
Bryan Andrews, Joseph Ramsey, Rub{\'e}n S{\'a}nchez-Romero, Jazmin Camchong,
  and Erich Kummerfeld.
\newblock Fast scalable and accurate discovery of {DAG}s using the best order
  score search and grow shrink trees.
\newblock In \emph{Proceedings of the Conference on Advances in Neural
  Information Processing Systems}, volume~36, pages 63945--63956, 2023.

\bibitem[Barab{\'a}si and Albert(1999)]{barabasi1999emergence}
Albert-L{\'a}szl{\'o} Barab{\'a}si and R{\'e}ka Albert.
\newblock Emergence of scaling in random networks.
\newblock \emph{Science}, 286:\penalty0 509--512, 1999.

\bibitem[Bello et~al.(2022)Bello, Aragam, and Ravikumar]{bello2022dagma}
Kevin Bello, Bryon Aragam, and Pradeep Ravikumar.
\newblock {DAGMA}: Learning {DAG}s via {M}-matrices and a log-determinant
  acyclicity characterization.
\newblock In \emph{Proceedings of the Conference on Advances in Neural
  Information Processing Systems}, volume~35, pages 8226--8239, 2022.

\bibitem[Bollen(1989)]{bollen1989structural}
Kenneth~A Bollen.
\newblock \emph{Structural equations with latent variables}, volume 210.
\newblock John Wiley \& Sons, 1989.

\bibitem[Chickering(2002)]{chickering2002optimal}
David~Maxwell Chickering.
\newblock Optimal structure identification with greedy search.
\newblock \emph{Journal of machine learning research}, 3:\penalty0 507--554,
  2002.

\bibitem[Colombo et~al.(2014)Colombo, Maathuis, et~al.]{colombo2014order}
Diego Colombo, Marloes~H Maathuis, et~al.
\newblock Order-independent constraint-based causal structure learning.
\newblock \emph{Journal of Machine Learning Research}, 15:\penalty0 3741--3782,
  2014.

\bibitem[Csardi and Nepusz(2005)]{igraph}
Gabor Csardi and Tamas Nepusz.
\newblock The igraph software package for complex network research.
\newblock \emph{InterJournal}, Complex Systems:\penalty0 1695, 11 2005.

\bibitem[Dawid(1979)]{dawid1979conditional}
A.~P. Dawid.
\newblock Conditional independence in statistical theory.
\newblock \emph{Journal of the Royal Statistical Society: Series B
  (Methodological)}, 41:\penalty0 1--15, 1979.

\bibitem[Drton(2018)]{drton2018algebraic}
Mathias Drton.
\newblock Algebraic problems in structural equation modeling.
\newblock In \emph{The 50th anniversary of Gr{\"o}bner bases}, volume~77 of
  \emph{Advanced Studies in Pure Mathematics}, pages 35--87. Mathematical
  Society of Japan, 2018.

\bibitem[Eberhardt(2017)]{Eberhardt2017-wi}
Frederick Eberhardt.
\newblock Introduction to the foundations of causal discovery.
\newblock \emph{International Journal of Data Science and Analytics},
  3:\penalty0 81--91, 2017.

\bibitem[Erd{\H o}s and R{\'e}nyi(1959)]{erdos59random}
Paul Erd{\H o}s and Alfr{\'e}d R{\'e}nyi.
\newblock On random graphs {I}.
\newblock \emph{Publicationes Mathematicae Debrecen}, 6:\penalty0 290--297,
  1959.

\bibitem[Fang et~al.(1990)Fang, Kotz, and Ng]{fang1990symmetric}
Kai-Tang Fang, Samuel Kotz, and Kai~W Ng.
\newblock \emph{Symmetric multivariate and related distributions}.
\newblock Chapman and Hall, 1990.

\bibitem[Geiger and Heckerman(2002)]{geiger2002parameter}
Dan Geiger and David Heckerman.
\newblock Parameter priors for directed acyclic graphical models and the
  characterization of several probability distributions.
\newblock \emph{The Annals of Statistics}, 30:\penalty0 1412--1440, 2002.

\bibitem[Ghosh and Henderson(2003)]{ghosh2003behavior}
Soumyadip Ghosh and Shane~G Henderson.
\newblock Behavior of the {NORTA} method for correlated random vector
  generation as the dimension increases.
\newblock \emph{ACM Transactions on Modeling and Computer Simulation},
  13:\penalty0 276--294, 2003.

\bibitem[Ghosh and Henderson(2009)]{ghosh2009corrigendum}
Soumyadip Ghosh and Shane~G Henderson.
\newblock Corrigendum: Behavior of the {NORTA} method for correlated random
  vector generation as the dimension increases.
\newblock \emph{ACM Transactions on Modeling and Computer Simulation},
  19:\penalty0 1--3, 2009.

\bibitem[Glymour(2007)]{glymour2007learning}
Clark Glymour.
\newblock {Learning the Structure of Deterministic Systems}.
\newblock In \emph{{Causal Learning: Psychology, Philosophy, and Computation}},
  pages 231--240. Oxford University Press, 2007.

\bibitem[Glymour et~al.(2019)Glymour, Zhang, and Spirtes]{Glymour2019-ce}
Clark Glymour, Kun Zhang, and Peter Spirtes.
\newblock Review of causal discovery methods based on graphical models.
\newblock \emph{Frontiers in Genetics}, 10:\penalty0 524, 2019.

\bibitem[Haughton(1988)]{haughton1988choice}
Dominique~MA Haughton.
\newblock On the choice of a model to fit data from an exponential family.
\newblock \emph{The annals of statistics}, pages 342--355, 1988.

\bibitem[Ikeuchi et~al.(2023)Ikeuchi, Ide, Zeng, Maeda, and
  Shimizu]{ikeuchi2023python}
Takashi Ikeuchi, Mayumi Ide, Yan Zeng, Takashi~Nicholas Maeda, and Shohei
  Shimizu.
\newblock Python package for causal discovery based on {LiNGAM}.
\newblock \emph{Journal of Machine Learning Research}, 24:\penalty0 1--8, 2023.

\bibitem[Jack~Kuipers and Moffa(2022)]{kuipers2022efficient}
Polina~Suter Jack~Kuipers and Giusi Moffa.
\newblock Efficient sampling and structure learning of {B}ayesian networks.
\newblock \emph{Journal of Computational and Graphical Statistics},
  31:\penalty0 639--650, 2022.

\bibitem[Kaiser and Sipos(2022)]{kaiser2022unsuitability}
Marcus Kaiser and Maksim Sipos.
\newblock Unsuitability of {NOTEARS} for causal graph discovery when dealing
  with dimensional quantities.
\newblock \emph{Neural Processing Letters}, 54:\penalty0 1587--1595, 2022.

\bibitem[Kalisch and B{\"u}hlman(2007)]{kalisch2007estimating}
Markus Kalisch and Peter B{\"u}hlman.
\newblock Estimating high-dimensional directed acyclic graphs with the
  {PC}-algorithm.
\newblock \emph{Journal of Machine Learning Research}, 8, 2007.

\bibitem[Khalafi and Azimmohseni(2014)]{khalafi2014multivariate}
Mohammad Khalafi and Majid Azimmohseni.
\newblock Multivariate {P}earson type {II} distribution: Statistical and
  mathematical features.
\newblock \emph{Probability and Mathematical Statistics}, 34:\penalty0
  119--126, 2014.

\bibitem[Kuipers et~al.(2014)Kuipers, Moffa, and
  Heckerman]{kuipers2014addendum}
Jack Kuipers, Giusi Moffa, and David Heckerman.
\newblock Addendum on the scoring of {G}aussian directed acyclic graphical
  models.
\newblock \emph{The Annals of Statistics}, pages 1689--1691, 2014.

\bibitem[Kummerfeld et~al.(2023)Kummerfeld, Williams, and
  Ma]{kummerfeld2023power}
Erich Kummerfeld, Leland Williams, and Sisi Ma.
\newblock Power analysis for causal discovery.
\newblock \emph{International Journal of Data Science and Analytics}, pages
  1--16, 2023.

\bibitem[Lam et~al.(2022)Lam, Andrews, and Ramsey]{lam2022greedy}
Wai-Yin Lam, Bryan Andrews, and Joseph Ramsey.
\newblock Greedy relaxations of the sparsest permutation algorithm.
\newblock In \emph{Proceedings of the Conference on Uncertainty in Artificial
  Intelligence}, pages 1052--1062. PMLR, 2022.

\bibitem[Lauritzen et~al.(1990)Lauritzen, Dawid, Larsen, and
  Leimer]{lauritzen1990independence}
Steffen~L Lauritzen, A~Philip Dawid, Birgitte~N Larsen, and H-G Leimer.
\newblock Independence properties of directed {M}arkov fields.
\newblock \emph{Networks}, 20:\penalty0 491--505, 1990.

\bibitem[Lewandowski et~al.(2009)Lewandowski, Kurowicka, and
  Joe]{lewandowski2009generating}
Daniel Lewandowski, Dorota Kurowicka, and Harry Joe.
\newblock Generating random correlation matrices based on vines and extended
  {O}nion method.
\newblock \emph{Journal of multivariate analysis}, 100:\penalty0 1989--2001,
  2009.

\bibitem[Malinsky and Danks(2018)]{Malinsky2018-mi}
Daniel Malinsky and David Danks.
\newblock Causal discovery algorithms: A practical guide.
\newblock \emph{Philosophy Compass}, 13:\penalty0 e12470, 2018.

\bibitem[Melan{\c{c}}on et~al.(2001)Melan{\c{c}}on, Dutour, and
  Bousquet-M{\'e}lou]{melanccon2001random}
Guy Melan{\c{c}}on, Isabelle Dutour, and Mireille Bousquet-M{\'e}lou.
\newblock Random generation of directed acyclic graphs.
\newblock \emph{Electronic Notes in Discrete Mathematics}, 10:\penalty0
  202--207, 2001.

\bibitem[Muller(1956)]{muller1956some}
Mervin~E Muller.
\newblock Some continuous {M}onte {C}arlo methods for the {D}irichlet problem.
\newblock \emph{The Annals of Mathematical Statistics}, 27:\penalty0 569--589,
  1956.

\bibitem[Ng et~al.(2024)Ng, Huang, and Zhang]{ng2024structure}
Ignavier Ng, Biwei Huang, and Kun Zhang.
\newblock Structure learning with continuous optimization: A look and beyond.
\newblock In \emph{Causal Learning and Reasoning}, pages 71--105. PMLR, 2024.

\bibitem[Ouellette(1981)]{ouellette1981schur}
Diane~Val{\'e}rie Ouellette.
\newblock Schur complements and statistics.
\newblock \emph{Linear Algebra and its Applications}, 36:\penalty0 187--295,
  1981.

\bibitem[Pearl(1988)]{pearl1988probabilistic}
Judea Pearl.
\newblock \emph{Probabilistic reasoning in intelligent systems: networks of
  plausible inference}.
\newblock Morgan kaufmann, 1988.

\bibitem[Pearl(2020)]{Pearl2020-xc}
Judea Pearl.
\newblock \emph{Book of Why}.
\newblock Basic Books, 2020.

\bibitem[Peters and B{\"u}hlmann(2014)]{peters2014identifiability}
Jonas Peters and Peter B{\"u}hlmann.
\newblock Identifiability of {G}aussian structural equation models with equal
  error variances.
\newblock \emph{Biometrika}, 101:\penalty0 219--228, 2014.

\bibitem[Price(1976)]{price1976general}
Derek de~Solla Price.
\newblock A general theory of bibliometric and other cumulative advantage
  processes.
\newblock \emph{Journal of the American society for Information science},
  27:\penalty0 292--306, 1976.

\bibitem[Ramsey et~al.(2017)Ramsey, Glymour, S{\'a}nchez-Romero, and
  Glymour]{ramsey2017million}
Joseph Ramsey, Madelyn Glymour, Rub{\'e}n S{\'a}nchez-Romero, and Clark
  Glymour.
\newblock A million variables and more: the fast greedy equivalence search
  algorithm for learning high-dimensional graphical causal models, with an
  application to functional magnetic resonance images.
\newblock \emph{International journal of data science and analytics},
  3:\penalty0 121--129, 2017.

\bibitem[Ramsey and Andrews(2017)]{Ramsey2017-rt}
Joseph~D Ramsey and Bryan Andrews.
\newblock A comparison of public causal search packages on linear, {G}aussian
  data with no latent variables, 2017.

\bibitem[Ramsey et~al.(2018)Ramsey, Zhang, Glymour, Romero, Huang,
  Ebert-Uphoff, Samarasinghe, Barnes, and Glymour]{ramsey2018tetrad}
Joseph~D Ramsey, Kun Zhang, Madelyn Glymour, Ruben~Sanchez Romero, Biwei Huang,
  Imme Ebert-Uphoff, Savini Samarasinghe, Elizabeth~A Barnes, and Clark
  Glymour.
\newblock {TETRAD}---a toolbox for causal discovery.
\newblock In \emph{8th International Workshop on Climate Informatics}, 2018.

\bibitem[Ramsey et~al.(2020)Ramsey, Malinsky, and Bui]{Ramsey2020-wn}
Joseph~D Ramsey, Daniel Malinsky, and Kevin~V Bui.
\newblock algcomparison: Comparing the performance of graphical structure
  learning algorithms with {TETRAD}.
\newblock \emph{Journal of Machine Learning Research}, 21:\penalty0 1--6, 2020.

\bibitem[Reisach et~al.(2021)Reisach, Seiler, and Weichwald]{reisach2021beware}
Alexander Reisach, Christof Seiler, and Sebastian Weichwald.
\newblock Beware of the simulated {DAG}! causal discovery benchmarks may be
  easy to game.
\newblock In \emph{Proceedings of the Conference on Advances in Neural
  Information Processing Systems}, volume~34, pages 27772--27784, 2021.

\bibitem[Reisach et~al.(2023)Reisach, Tami, Seiler, Chambaz, and
  Weichwald]{reisach2023simple}
Alexander Reisach, Myriam Tami, Christof Seiler, Antoine Chambaz, and Sebastian
  Weichwald.
\newblock A scale-invariant sorting criterion to find a causal order in
  additive noise models.
\newblock In \emph{Advances in Neural Information Processing Systems},
  volume~36, pages 785--807, 2023.

\bibitem[Rios et~al.(2021)Rios, Moffa, and Kuipers]{rios2021benchpress}
Felix~L. Rios, Giusi Moffa, and Jack Kuipers.
\newblock Benchpress: a scalable and versatile workflow for benchmarking
  structure learning algorithms for graphical models, 2021.

\bibitem[Schwarz(1978)]{schwarz1978estimating}
Gideon Schwarz.
\newblock Estimating the dimension of a model.
\newblock \emph{The annals of statistics}, pages 461--464, 1978.

\bibitem[Shimizu et~al.(2006)Shimizu, Hoyer, Hyv{\"a}rinen, Kerminen, and
  Jordan]{shimizu2006linear}
Shohei Shimizu, Patrik~O Hoyer, Aapo Hyv{\"a}rinen, Antti Kerminen, and Michael
  Jordan.
\newblock A linear non-{G}aussian acyclic model for causal discovery.
\newblock \emph{Journal of Machine Learning Research}, 7, 2006.

\bibitem[Shimizu et~al.(2011)Shimizu, Inazumi, Sogawa, Hyvarinen, Kawahara,
  Washio, Hoyer, Bollen, and Hoyer]{shimizu2011directlingam}
Shohei Shimizu, Takanori Inazumi, Yasuhiro Sogawa, Aapo Hyvarinen, Yoshinobu
  Kawahara, Takashi Washio, Patrik~O Hoyer, Kenneth Bollen, and Patrik Hoyer.
\newblock Direct{LiNGAM}: A direct method for learning a linear non-{G}aussian
  structural equation model.
\newblock \emph{Journal of Machine Learning Research}, 12:\penalty0 1225--1248,
  2011.

\bibitem[Spirtes(2010)]{Spirtes2010-cg}
Peter Spirtes.
\newblock Introduction to causal inference.
\newblock \emph{Journal of Machine Learning Research}, 11:\penalty0 1643--1662,
  2010.

\bibitem[Spirtes and Zhang(2014)]{Spirtes2014-au}
Peter Spirtes and Jiji Zhang.
\newblock A uniformly consistent estimator of causal effects under the
  {$k$-triangle-faithfulness} assumption.
\newblock \emph{Statistical Science}, pages 662--678, 2014.

\bibitem[Spirtes et~al.(1993)Spirtes, Glymour, and Scheines]{Spirtes1993-od}
Peter Spirtes, Clark Glymour, and Richard Scheines.
\newblock \emph{Causation, Prediction, and Search}.
\newblock Lecture Notes in Statistics. Springer, 1993.

\bibitem[Suter et~al.(2023)Suter, Kuipers, Moffa, and Beerenwinkel]{bidag}
Polina Suter, Jack Kuipers, Giusi Moffa, and Niko Beerenwinkel.
\newblock {B}ayesian structure learning and sampling of {B}ayesian networks
  with the {R} package {BiDAG}.
\newblock \emph{Journal of Statistical Software}, 105:\penalty0 1--31, 2023.

\bibitem[Tsagris(2021)]{tsagris2021new}
Michail Tsagris.
\newblock A new scalable {B}ayesian network learning algorithm with
  applications to economics.
\newblock \emph{Computational Economics}, 57:\penalty0 341--367, 2021.

\bibitem[Tsagris(2023)]{pchc}
Michail Tsagris.
\newblock \emph{pchc: {B}ayesian Network Learning with the {PCHC} and Related
  Algorithms}, 2023.
\newblock URL \url{https://CRAN.R-project.org/package=pchc}.
\newblock R package version 1.2.

\bibitem[Tsamardinos et~al.(2006)Tsamardinos, Brown, and
  Aliferis]{tsamardinos2006max}
Ioannis Tsamardinos, Laura~E Brown, and Constantin~F Aliferis.
\newblock The max-min hill-climbing {B}ayesian network structure learning
  algorithm.
\newblock \emph{Machine learning}, 65:\penalty0 31--78, 2006.

\bibitem[Uhler et~al.(2013)Uhler, Raskutti, Bühlmann, and
  Yu]{uhler2013geometry}
Caroline Uhler, Garvesh Raskutti, Peter Bühlmann, and Bin Yu.
\newblock Geometry of the faithfulness assumption in causal inference.
\newblock \emph{The Annals of Statistics}, 41:\penalty0 436--463, 2013.

\bibitem[Verma and Pearl(1990)]{verma1990causal}
Thomas Verma and Judea Pearl.
\newblock Causal networks: Semantics and expressiveness”.
\newblock In \emph{Uncertainty in Artificial Intelligence}, volume~9 of
  \emph{Machine Intelligence and Pattern Recognition}, pages 69--76.
  North-Holland, 1990.

\bibitem[Wang and Spirtes(2022)]{pmlr-v177-wang22a}
Shuyan Wang and Peter Spirtes.
\newblock A uniformly consistent estimator of non-{G}aussian causal effects
  under the $k$-triangle-faithfulness assumption.
\newblock In Bernhard Sch{\"o}lkopf, Caroline Uhler, and Kun Zhang, editors,
  \emph{Proceedings of the Conference on Causal Learning and Reasoning}, volume
  177 of \emph{Proceedings of Machine Learning Research}, pages 861--876. PMLR,
  2022.

\bibitem[Zheng et~al.(2018)Zheng, Aragam, Ravikumar, and Xing]{zheng2018dags}
Xun Zheng, Bryon Aragam, Pradeep~K Ravikumar, and Eric~P Xing.
\newblock {DAG}s with no tears: Continuous optimization for structure learning.
\newblock In \emph{Proceedings of the Conference on Advances in Neural
  Information Processing Systems}, volume~31, pages 9492--9503, 2018.

\end{thebibliography}


\end{document}